\documentclass{IEEEtran}
\usepackage[ansinew]{inputenc}
\usepackage{epsfig}
\usepackage[T1]{fontenc}
\usepackage{amsmath,enumerate,amssymb,hhline,verbatim}
\usepackage{xcolor}
\usepackage[square, comma, numbers]{natbib}

\newenvironment{proof}{\begin{IEEEproof}}{\end{IEEEproof}}

\newcommand{\bbar}[1]{\setbox0=\hbox{$#1$}\dimen0=.2\ht0 \kern\dimen0 \overline{\kern-\dimen0 #1}}

\newcommand{\Z}{{\mathbb Z}}
\newcommand{\ZZ}{{\mathbb Z}}
\newcommand{\D}{{\mathcal D}}

\newcommand{\A}{{\mathcal A}}

\newcommand{\I}{{\mathcal I}}

\newcommand{\mindet}[1]{\hbox{\rm det}_{min}\left( #1\right)}
\newcommand{\diag}{{\rm diag}}
\newcommand{\R}{{\mathbb R}}
\newcommand{\C}{{\mathbb C}}

\newcommand{\OO}{{\mathcal O}}
\newcommand{\Q}{{\mathbb Q}}

\newtheorem{thm}{Theorem}[section]
\newtheorem{theorem}[thm]{Theorem}

\newtheorem{cor}[thm]{Corollary}
\newtheorem{corollary}[thm]{Corollary}
\newtheorem{lem}[thm]{Lemma}
\newtheorem{lemma}[thm]{Lemma}
\newtheorem{proposition}[thm]{Proposition}
\newtheorem{prop}[thm]{Proposition}

\newtheorem{defn}[thm]{Definition}
\newtheorem{definition}[thm]{Definition}

\newtheorem{remark}[thm]{Remark}

\providecommand{\abs}[1]{\ensuremath{\left\lvert #1 \right\rvert}}

\providecommand{\norm}[1]{\ensuremath{\left\Vert #1 \right\Vert}}

\providecommand{\vv}[1]{\textquotedblleft #1\textquotedblright}

\DeclareMathOperator*{\Vol}{Vol}
\DeclareMathOperator*{\tr}{tr}
\DeclareMathOperator*{\argmin}{argmin}
\DeclareMathOperator*{\SNR}{SNR}
\renewcommand{\IEEEQED}{\IEEEQEDopen}

\begin{document}

\title{ Almost universal codes achieving ergodic MIMO capacity within a constant gap}

\author{Laura Luzzi and Roope Vehkalahti

\thanks{The research of R.~Vehkalahti was funded by   Academy of Finland  grant  \#252457 and is funded by the Finnish Cultural Foundation.}
\thanks{Part of this work appeared at ISIT 2015 \cite{ISIT2015_SISO, ISIT2015_MIMO}.}
\thanks{L. Luzzi is with ETIS (UMR 8051, ENSEA, Universit\'e de Cergy-Pontoise, CNRS), 95014 Cergy-Pontoise, France (e-mail: laura.luzzi@ensea.fr).}
\thanks{R.~Vehkalahti is with the Department of Mathematics and Statistics, FI-20014, University of Turku, Finland  (e-mail: roiive@utu.fi)}
}

\maketitle

\begin{abstract}
This work addresses the question of achieving capacity with lattice codes in multi-antenna block fading channels when the number of fading blocks tends to infinity. \\
A design criterion based on the normalized minimum determinant is proposed for division algebra multiblock space-time codes over fading channels; this plays a similar role to the Hermite invariant for Gaussian channels. \\ 
It is shown that this criterion is sufficient to guarantee transmission rates within a constant gap from capacity both for slow fading channels and ergodic fading channels. This performance is achieved both under maximum likelihood decoding and naive lattice decoding.
In the case of independent identically distributed Rayleigh fading, it is also shown that the error probability vanishes exponentially fast.\\
In contrast to the standard approach in the literature which employs random lattice ensembles, the existence results in this paper are derived from number theory. First the gap to capacity is shown to depend on the discriminant of the chosen division algebra; then class field theory is applied to build families of algebras with small discriminants. The key element in the construction is the choice of a sequence of division algebras whose centers are
number fields with small root discriminants.
\end{abstract}


\section{Introduction}
It is well-known \cite{Tel} that in ergodic multiple-input multiple-output (MIMO) fading channels with channel state information at receiver only, the maximal mutual information is achieved with Gaussian circularly symmetric random inputs.  In this case the existence of capacity-achieving codes can be proven with standard random coding arguments.  \\
 It has been shown that by combining simple modulation and strong outer codes such as turbo or LDPC codes, it is possible to operate at rates close to capacity with small error probability  \cite{HoBr,Sanderovich}. However, to the best of our knowledge, the problem of achieving capacity with \emph{explicit codes} for all ranges of signal-to-noise ratio (SNR) is still open.\\ 
This is in strong contrast to the classical complex Gaussian single antenna channel, where the capacity is $\log(1+\SNR)$ and it is known that several lattice code constructions achieve $\log \SNR -C$ rates. These constructions are based on  a rich theory of lattice codes developed to attack these questions. At the heart of this theory are sphere packing arguments that prove that the performance of a lattice code in the classical Gaussian channel can be roughly estimated by the size of a geometrical invariant of the lattice, the \emph{Hermite invariant}.  In particular the Hermite invariant can be used to roughly measure how close to capacity a family of lattices can get.
	This connection has been extremely fruitful and has led to a monumental work connecting algebra, geometry and information theory \cite{CS}.  

In the case of fading channels the situation is quite different.
While it is well-known that space-time lattice codes from division algebras \cite{basic} provide good performance over multiple antenna fading channels, 
and a rich algebraic theory has been developed to 
optimize single codes 
\citep{ORBV}, there are as yet no results connecting capacity questions and the geometry of lattices.
 The minimum determinant criterion \cite{TSC} 
 allows to improve the worst-case pairwise error probability in the high-SNR regime, when coding over a  single fading block. Optimizing this value has been the major concern of several works in space-time coding \cite{ORBV, WX,VHLR}.
 However, no  design criterion has been suggested for approaching the MIMO capacity with explicit lattice codes.

In this paper we address this problem 
and show that when we are allowed to encode and decode over a growing number of fading blocks, 
the \emph{normalized minimum determinant} plays a similar role to the Hermite constant in Gaussian channels. In particular it can be used to measure how close to capacity a given family of lattice codes can get. 

Based on this design criterion we prove that for a MIMO channel with $n$ transmit and $n_r$ receive antennas, there exists a family of multiblock lattice codes $L_{n,k}\subset M_{n\times nk}(\C)$ (where $k$ goes to infinity) that achieves a constant gap to capacity both in the slow fading and ergodic fading case.  
More precisely, for a MIMO channel with ergodic capacity $C=\mathbb{E}_{H}\left[\log \det (I_{n_r} + \frac{\SNR}{n} H^{\dagger}H)\right]$, our scheme achieves any rate 
\begin{equation}\label{main}
R< \mathbb{E}_{H}\left[\log \det \frac{\SNR}{n} H^{\dagger}H\right]-n\log C_L +n\log\frac{\pi e}{4n},
\end{equation}
where $C_L$ is a certain geometric invariant of the family of lattices. These rates are achieved not only with maximum likelihood (ML) decoding, but also with naive lattice decoding. \\
Furthermore, the same scheme achieves positive rates of reliable communication for more general fading processes $\{H_i\}$ under the mild hypothesis that the weak law of large numbers holds for the sequence of random variables $\{\log \det H_i^{\dagger}H_i\}$. \\
As far as we know, this is the first  explicit coding scheme which achieves 
constant gap to capacity for all SNR levels in  MIMO channels.  
	
Instead of using random coding arguments we 
consider 
algebraic multi-block division algebra codes \cite{YB07, Lu}.
Our lattice constructions are based on two results from classical class field theory. First  we choose the center $K$ of the algebra from an ensemble of Hilbert class fields having small root discriminant  and then we prove the existence of a $K$-central division algebra with small discriminant. Our lattices   belong to a very general family of division algebra codes introduced in \cite{YB07, Lu, EK}, and developed further in  \cite{HL} and \cite{VHO}. We will use the most general form presented in \cite{LSV}.
In particular our work proves that the classical number field codes \cite{OV} achieve a constant gap to capacity in Rayleigh fast fading single antenna channels.

In most works on algebraic space-time coding the code design criterion is derived from an upper bound for the pairwise error probability \cite{TSC} together with the union bound \cite{OV,OBV}.
 In our proofs we abandon this method and consider a hard sphere packing approach, classically used in lattice coding for the AWGN channel.  The idea, formalized in  Section \ref{criterion}, is to exploit the special \emph{multiplicative structure} of algebraic codes.   It was observed in the context of diversity-multiplexing gain trade-off (DMT) analysis \cite{EKPKL, TV}, that fading has a diminishing effect on the euclidean distance of the received code constellations derived from division algebra codes, having the so-called \emph{non-vanishing determinant} property, only if the channel itself is bad. This property was formalized in \cite{TV}, where the authors introduced \emph{approximately universal codes} for fading channels.
Our main results Theorem  \ref{prop_positive_rate} and Theorem \ref{prop_positive_rate2} rely on this ``incompressibility'' property of algebraic lattices. It follows that our codes are  \emph{almost universal} and perform within a constant gap to capacity for a wide  class of channels, having only mild restrictions on fading.

While we discuss specific lattice codes from division algebras, our proofs do work for any ensemble of  matrix lattices with asymptotically good normalized minimum determinant.  The larger this value is, the smaller the gap to the capacity.

This work  also suggests  that capacity questions in fading channels are naturally linked to problems in the mathematical research area of \emph{geometry of numbers}. 
Unlike the single antenna Gaussian case, many of the questions that arise have not been actively studied by the mathematical community. \\
Hopefully, studying such questions may lead to a comprehensive geometric theory of lattices for multiple antenna fading channels.


We note that the proposed lattice code constructions are not yet practical, since they are based on number fields whose existence is proved through class field theory. Given a fixed degree, the required number fields can be found using computational algebra software, but this process is computationally taxing. Decoding of the proposed codes is also very complex and the constructions we provide still have a large gap to capacity.

On the other side, as demonstrated in Section \ref{numberfields}, the existence results we use are very pessimistic. For small 
degrees the normalized minimum determinants of the best possible lattices are considerably better than the bounds provided by our existence results.


\subsection{Related work}\label{related}
  While our work shows that one can achieve a constant gap to capacity in ergodic MIMO channels with a fixed family of algebraic codes, 
it is natural to consider the 
more general question of whether it is possible to achieve 
capacity with 
any 
lattice codes. Such a result would be 
a 
generalization 
of the work in \cite{DeBuda, Urbanke_Rimoldi, Loeliger} which proved the existence of random lattice code ensembles 
achieving rate $\log(\mathrm{SNR})$ over the AWGN channel. By making the extra assumption that the transmitter and receiver have access to a common source of randomness in the form of a dither, the authors in \cite{Erez_Zamir} finally proved that the AWGN capacity is achievable with random lattice codes. An explicit multilevel construction from polar codes was recently proposed in \cite{Yan_Ling_Wu}. 

As far as we know our work \cite{ISIT2015_MIMO} was the first to give a proof that lattice codes achieve a constant gap to capacity in block fading MIMO channels. 
In the single antenna fast fading channel this problem was considered before 
in \cite{HNGLOBECOM}, 
which claims 
that random lattices achieve a constant gap to capacity. In \cite{HNISIT} the authors extend their previous results and claim to give a proof that random lattices achieve capacity in single antenna ergodic fading channels. However, we believe that at least in its current form, the analysis in both works is missing some fundamental details. In particular,  
the gap 
$
\Delta <1 + \log\mathbb{E}_h\left[\frac{1}{|h|^2}\right],$
given in \cite[Theorem 3]{HNGLOBECOM}, is infinite even 
when the fading process $\{h_i\}$ is i.i.d. complex Gaussian.  
In \cite[Equation (20)]{HNISIT} the authors state that for a given fixed fading realization, the Minkowski-Hlawka theorem implies that there exist a lattice for which the error probability is upper bounded in a certain way. 
However, they 
proceed as if there existed a single lattice that would satisfy 
this upper bound for any channel state. To the best of our knowledge, such a result can not be derived from Minkowski-Hlawka. 
\smallskip \par
As far as explicit algebraic constructions are concerned, our work is indebted to several previous papers.\\ 
The idea to use division algebra codes to achieve capacity can be tracked down to the work of   H.-f. Lu  in \cite{Lu}.
While studying the diversity-multiplexing gain tradeoff (DMT) of  multiblock codes he   conjectured that the ensemble of  multi-block division algebra codes might approach the ergodic Rayleigh fading capacity. Our work confirms that conjecture; however, we point out that it is unlikely that DMT-optimality alone is enough to approach capacity. Instead one should pick the code very carefully by maximizing the normalized minimum determinant. 

The families of number fields  on which our constructions are based 
were first brought to coding theory in \cite{LT}, where the authors pointed out that the corresponding lattices  have large Hermite constant. 
C. Xing in \cite{Xing} remarked that   these families of number fields provide the best known normalized product distance making them a natural candidate for achieving constant gap to capacity in fading single antenna channels.

Our results  on slow fading channel are motivated by the work in \cite{OE}, where the authors prove that \emph{precoded integer forcing} achieves a constant gap to capacity for every slow fading channel with fixed fading. 

Our  geometry of   numbers approach has its roots in   \cite{GB}, where the authors studied lattice codes in single antenna fading channel and defined the normalized product distance. They also pointed out that using this criterion reduces the lattice design to a problem in geometry of numbers.

 The generalization of these ideas to the MIMO channel was developed in  \cite{LV} and \cite{V}, where the code design for quasi-static MIMO channel problem was translated into lattice theoretic language and where a formal definition of normalized minimum determinant was given. However, none of these works considered the relation between geometry of numbers and capacity problems.


\subsection{Organization of the paper}
In Section \ref{preliminaries} we introduce the multiblock channel model and recall the relevant properties of lattice codes. In Section \ref{reher} we develop a geometric design criterion for capacity approaching lattice codes for fading multiple antenna channels and define the concept of \emph{reduced Hermite invariant} which is an analogue of the classical Hermite invariant. In Section \ref{statement} we state the existence of lattices having asymptotically good normalized minimum determinant (the proof will be given in Section \ref{construction}).    
 In Section \ref{general_channels_section} we prove that the lattice codes of the previous section achieve positive rates over a very general class of channels. We then prove that they achieve a constant gap to capacity over slow fading (Section \ref{slow_fading_section}) and ergodic fading channels (Section \ref{ergodic_section}). In Section \ref{Rayleigh_fading} we focus on the i.i.d. Rayleigh fading channel model, and show that the error probability vanishes exponentially. In Section \ref{construction} we prove the existence of asymptotically good lattices, and in Section \ref{numberfields} we specialize our results to the single antenna case.  Finally in Section \ref{fading_geom} we explore the connection between capacity questions in fading channels and  geometry of numbers. Section \ref{conclusion} discusses some perspectives and open problems. 

\subsection{Notation}
Throughout the paper, capacity is measured in bits. Accordingly, we denote by $\log$ the base $2$ logarithm in rate and capacity expressions; the natural logarithm will be denoted by $\ln$.

\section{Multiblock lattice codes} \label{preliminaries}

\subsection{Channel model} \label{channel_model}
We consider a MIMO system with $n$ transmit and $n_r$ receive antennas, where transmission takes place over $k$ quasi-static fading blocks of  delay $T=n$. 
Each multi-block codeword $X\in M_{n\times nk}(\C)$ has the form $(X_1,X_2,\dots, X_k)$, where the submatrix $X_i \in M_n(\C)$ is sent during the $i$-th block. The received signals are given by
\begin{equation}\label{eq:channel}
Y_i=H_i X_i +W_i, \quad \quad i \in \{1,\ldots,k\}
\end{equation}
where $H_i \in M_{n_r \times n}(\C)$ and $W_i\in M_{n_r \times T}(\C)$ are the channel and noise matrices. The coefficients of $W_i$ are modeled as circular symmetric complex Gaussian with zero mean and unit variance per complex dimension.  Perfect channel state information is available at the receiver but not at the transmitter, and decoding is performed after all $k$ blocks have been received. We will call such a channel an \emph{$(n, n_r, k)$-multiblock channel}.\\
In this paper, we will also assume that for all $i \geq 1$, $H_i \in M_{n_r \times n}$ is full-rank with probability $1$, and that the random variable $\sum_{i=1}^k \frac{1}{n} \log \det (H_i^{\dagger}H_i)$ converges in probability to some constant when the number of blocks $k$ tends to infinity.  This channel model covers several standard MIMO channels such as the Rayleigh block fading channel and the slow fading channel.

A \emph{multi-block code}  $\mathcal{C}$ in a $(n, n_r, k)$-channel is  a set of matrices  in $M_{n\times nk}(\C)$.  In particular we will concentrate on finite codes that are drawn from lattices. Let $R$ denote the code rate in bits per complex channel use; equivalently, $\abs{\mathcal{C}}=2^{Rkn}$. 
We assume that every matrix $X$ in a finite code $\mathcal{C}\subset M_{n\times nk}(\C)$ satisfies the average power constraint
\begin{equation} \label{power_constraint}
\frac{1}{nk} \norm{X}^2 \leq P,
\end{equation}
where $\norm{X}$ is the Frobenius norm of the matrix $X$.


\subsection{Lattice codes} \label{lattice_basics}

\begin{definition}
A {\em matrix lattice} $L \subseteq M_{n\times nk}(\C)$ has the form
$$
L=\Z B_1\oplus \Z B_2\oplus \cdots \oplus \Z B_r,
$$
where the matrices $B_1,\dots, B_r$ are linearly independent over $\R$, i.e., form a lattice basis, and $r$ is
called the \emph{rank}  or the \emph{dimension} of the lattice.
\end{definition}

The space $M_{n\times nk}(\C)$ is a $2n^2k$-dimensional real vector space with a real inner product
$$
\langle X,Y\rangle=\Re(Tr (XY^{\dagger})),
$$
where $Tr$ is the matrix trace. This inner product also naturally defines a metric on the space $M_{n\times nk}(\C)$ by setting $||X||= \sqrt{\langle X,X\rangle}$. 
 

Given an  $m$ dimensional lattice $L\subset M_{n\times nk}(\C)$, its 
{\it Gram matrix} 
is defined as 
$$G(L)=\left( \langle X_i,X_j\rangle\right)_{1\le i,j\le m},$$
where $\{ X_i \}_{1 \leq i \leq m}$ is a basis of $L$.  The volume of the fundamental parallelotope of $L$ is then defined as 
$\Vol(L)=\sqrt{|\det(G(L))|}$.

In the following we will use the notation $\R(L)$ for the linear space generated by the basis elements of the lattice $L$.
\begin{lem}\emph{\cite{GL}}\label{shift}
Let us suppose that $L$ is a lattice in  $M_{n\times kn}(\C)$  and  $S$ is a Jordan measurable bounded subset of $\R(L)$.  Then there exists $X \in M_{n\times kn}(\C)$ such that
$$
|(L+X)\cap S|\geq\frac{\mathrm{Vol}(S)}{\mathrm{Vol}(L)}.
$$
\end{lem}

Given a family of lattices $L_{n,k}\subseteq M_{n\times nk}(\C)$, let us now show how we can design multiblock codes $\mathcal{C}$ having rate greater or equal to a prescribed constant $R$, and satisfying the average power constraint (\ref{power_constraint}), from a scaled version $\alpha L_{n,k}$ of the lattices, where $\alpha$ is a suitable energy normalization constant. We denote by $B(r)$ the set of matrices in $M_{n \times nk}(\C)$ with Frobenius norm smaller or equal to $r$. According to Lemma \ref{shift}, we can choose a constant shift $X_R\in M_{n\times nk}(\C)$ such that for $\mathcal{C}=B(\sqrt{Pkn} )\cap(X_R+\alpha L_{n,k})$ we have
$$2^{Rnk}= \abs{\mathcal{C}}\geq\frac{\Vol(B(\sqrt{Pkn} ))}{\Vol(\alpha L_{n,k})}=\frac{C_{n,k}P^{n^2k}}{\alpha^{2n^2k} \Vol(L_{n,k})},$$
where $C_{n,k}=\frac{(\pi n k)^{n^2 k}}{(n^2 k)!}$. We then find the following condition for the scaling constant:
\begin{equation} \label{alpha_multiblock}
\alpha^2=\frac{C_{n,k}^{\frac{1}{n^2k}} P}{2^{\frac{R}{n}} \Vol(L_{n,k})^{\frac{1}{n^2 k}}} 
\end{equation}

\section{Design criteria for fading channels}\label{criterion}

In this section we 
propose a new design criterion for capacity approaching lattice codes in fading channels. For the sake of simplicity, we will focus on the case $n_r \geq n$.
We note that 
the design criterion derived here will finally be the familiar minimum determinant criterion. However, we hope that our alternative characterization offers more insight on the topic and can have applications in further research.

\subsection{Reduced Hermite invariant}\label{reher}

We recall the classical definition of the Hermite invariant, which characterizes the density of a lattice packing:
\begin{definition}\label{def:hermite}
The Hermite invariant of an $m$-dimensional  lattice $L\subset M_{n\times nk}(\C)$  can be defined as 
$$
\mathrm{h}(L)=\frac{\mathrm{inf}\{\,||X||^2 \mid  X\in L, X\neq 0\}}{\Vol(L)^{2/m}}.
$$
\end{definition}
On the $n \times n$ MIMO Gaussian channel such that the channel matrices $H_i=I_n$ $\forall i$, the classical sphere packing approach 
is to choose a $2n^2k$-dimensional lattice code $L_{n, k}\subset M_{n\times nk}(\C)$ such that $\mathrm{h}(L_{n,k})$ is as large as possible. 

Let us now assume that we have a finite code $\mathcal{C}_L \subset L_{n,k}\subset M_{n\times nk}(\C)$ and a random channel realization $H=[H_1,\dots, H_k]$. 
We note that the channel output $Y=[Y_1,\ldots,Y_k]$ can be written as 
$$Y=H_d X + W,$$
where $X=[X_1,\ldots,X_k]$, $W=[W_1,\ldots,W_k]$, and $H_d=\diag(H_1,\ldots,H_k)$. From the receiver's point of view, this is equivalent to an additive white Gaussian noise channel where the lattice code is
$$
H\mathcal{C}_L=\{H_dX\mid X\in \mathcal{C}_L\}.
$$
Even if the lattice $L_{n,k}$ (and therefore the code $\mathcal{C}_L$) has good minimum distance, there is no guarantee that the same can be said about the lattice $H\mathcal{C}_L$. 
This leads us to consider matrix lattices $L_{n,k}\subset M_{n\times nk}(\C)$ which would have good minimum distance after any (reasonable) channel. If we assume that each of the matrices $H_i$ in equation \eqref{eq:channel} has full rank with probability $1$, then the multiplication $X \mapsto H_d X$ is a bijective linear mapping with probability $1$. For any lattice $L_{n,k}\subset M_{n\times nk}(\C)$ having basis $B_1,\dots ,B_{2n^2k}$ we then have that
$$
HL_{n,k}=\{H_dX\mid X\in L_{n,k}\}=\Z H_dB_1 \oplus \cdots \oplus \Z H_dB_{2n^2k},
$$
is a lattice with basis $H_dB_1,\cdots, H_dB_{2n^2k}$, and
$\mathrm{h}(HL_{n,k})$ is well defined.

 As a discrete group, $HL_{n,k}$  has positive Hermite invariant, but even if $\mathrm{h}(L_{n,k}) $ is large there is no guarantee that $\mathrm{h}(HL_{n,k})$ is. 

For convenience we first introduce a group of matrices
\begin{equation} \label{G_definition}
G= \{ H\in M_{n\times nk }(\C  )\mid\mathrm{pdet}(H)=1\},
\end{equation}
where $\mathrm{pdet}(H)=\prod_{i=1}^{k} \mathrm{det}(H_i)$.

\begin{definition}\label{reduced}
The  \emph{reduced Hermite invariant} of an $m$-dimensional  lattice $L\subset M_{n\times nk}(\C)$ with respect to the group $G$  is defined as 
$$
\mathrm{rh_G}(L)=\underset{H\in G}{\mathrm{inf}} \{\mathrm{h}(HL)\}.
$$
\end{definition}

For any lattice $L$, $\mathrm{h}(L)>0$. The same is not true for the reduced Hermite invariant.
Let us now describe the set of lattices $L$ for which $\mathrm{rh_G}(L)>0$.

\begin{definition}\label{mindet}
The \emph{minimum determinant} of the lattice $L \subseteq M_{n\times nk}(\C)$ is defined as
\[
\mindet{L}:=\inf_{X \in L \setminus \{\bf 0\}} \abs{\mathrm{pdet}(X)}.
\]
If $\mindet{L}>0$ we say that the lattice satisfies the \emph{non-vanishing determinant} (NVD) property.
\end{definition}

We can now define the {\em normalized minimum determinant} $\delta(L)$, which  is obtained  by first scaling the lattice $L$ to have a unit size fundamental parallelotope and then taking the minimum determinant of the resulting scaled lattice. A simple computation proves the following.
\begin{lemma}\label{scale}
Let $L$ be an $m$-dimensional matrix lattice
in $M_{n\times nk}(\C)$. We then have that
\begin{equation} \label{normalized_minimum_determinant}
\delta(L) =\frac{\mindet{L}}{(\Vol(L))^{nk/m}}.
\end{equation}
\end{lemma}

The normalized minimum determinant provides an alternative characterization of the reduced Hermite invariant, but before that we need a well known lemma.
\begin{lemma} \label{hadamard}
Let $A$ be an $m\times m$ complex matrix. We  have the inequality
$$
\vert\det (A)\vert\le\frac{\Vert A \Vert ^m}{m^{m/2}}.
$$
\end{lemma}
\vspace* {4pt}
For a matrix $X \in M_{n\times nk}(\C)$ this immediately implies  that 
$$
|\mathrm{pdet}(X)|\le\frac{\Vert X \Vert^{nk}}{(nk)^{nk/2}}.
$$

\begin{proposition}\label{motivation}
If  $L\subset M_{n\times nk}(\C)$ is a $2n^2k$-dimensional lattice, then 
$$
  nk \left( \delta(L)\right)^{2/nk}=\mathrm{rh_G}(L).
$$
\end{proposition}
\begin{proof}
If the lattice $L$ includes a non-zero element $X$ such that $\mathrm{pdet}(X)= 0$, it is easy to see that $nk \left( \delta(L)\right)^{2/nk}=\mathrm{rh_G}(L)=0$.\\
Let us now assume that $\mathrm{pdet}(X)\neq 0$, for all $X\neq 0$. If $\mathrm{pdet}(H)=1$, Lemma \ref{hadamard} implies that
$$
\norm{H_dX}^2\geq nk\abs{\mathrm{pdet}(H_dX)}^{2/nk}=nk\abs{\mathrm{pdet}(X)}^{2/nk}.
$$
It follows that $nk \left( \delta(L)\right)^{2/nk}\leq \mathrm{rh_G}(L)$. 

Let us now assume that we have a sequence of  codewords $X^{(i)} \in L$ such that
$$
\lim_{i\to \infty}  nk |\mathrm{pdet}(X^{(i)})|^{2/nk}=nk \left( \delta(L)\right)^{2/nk}.
$$
If $X^{(i)}=[X^{(i)}_1,\dots, X^{(i)}_k]$, 
we can choose 
$H^{(i)}=\mathrm{pdet}(X^{(i)})^{1/n} [(X^{(i)}_1)^{-1},\dots, (X^{(i)}_k)^{-1}]$ 
so that $\mathrm{pdet}(H^{(i)})=1$. We then have
\begin{align*}
&||H_d^{(i)}X^{(i)}||^2=nk|\mathrm{pdet}(H_d^{(i)}X^{(i)})|^{2/nk}=\\
&=nk|\mathrm{pdet}(X^{(i)})|^{2/nk} 
\end{align*}
for every $i$ and therefore $$\lim_{i\to \infty } ||H_d^{(i)}X^{(i)}||^2=nk \left( \delta(L)\right)^{2/nk}.$$ It follows that
$nk \left( \delta(L)\right)^{2/nk}= \mathrm{rh_G}(L)$.
\end{proof}

\begin{remark}
Our definition of the reduced Hermite invariant $\mathrm{rh}_G$ depends heavily on the group $G$.  The group chosen in (\ref{G_definition})
can be seen as a block diagonal subgroup of $\mathrm{SL}_{kn}(\C)$. We could also consider a subgroup $G_1\subset G$ and define $\mathrm{rh}_{G_1}$ with respect to this group.
A natural consequence of these definitions is that for  two subgroups $G_1$, $G_2$  of $G$ such that $G_1\subseteq G_2$, we have that
$$
\mathrm{rh}_{G_1}(L)\geq \mathrm{rh}_{G_2}(L).
$$
 \end{remark}

\subsection{Asymptotically good families of lattices}\label{statement}
Based on the observations in the previous section, we can say that a sequence of codes $L_{n,k}$ is 
\emph{asymptotically good for the AWGN channel} if
$\mathrm{h}(L_{n,k})\geq cn^2k$, for some positive fixed constant $c$. Similarly, we can say that a sequence of lattices is \emph{asymptotically good for fading channels} if $\mathrm{rh_G}(L_{n,k})\geq cn^2k$. As seen in Proposition \ref{motivation}, this is equivalent to asking that $\delta(L_{n,k})^{2/nk} \geq cn$.

In order to keep the paper suitable for a larger audience we will postpone the proof of the following existence result to Section \ref{construction}.

\begin{proposition} \label{prop:volume}
Given $n$, there exists a family of $2n^2k$-dimensional lattices $L_{n,k}\subset M_{n\times nk}(\C)$, where $k$ grows to infinity, and a constant $G< 92.4$ such that
\begin{align*}
&\Vol(L_{n,k}) \leq 23^{\frac{kn(n-1)}{10}}\left(\frac{G}{2}\right)^{n^2k} \,\\
& {\det}_{min}(L_{n,k})=1\quad \mathrm{and} \quad \delta(L_{n,k})\geq \frac{1}{23^{\frac{k}{20}(n-1)}(G/2)^{\frac{nk}{2}}}.
\end{align*}
\end{proposition}

\begin{remark}
In this section we have developed the notion of reduced Hermite invariant for the case $n_r\geq n$. We observe that this notion does not extend in a straightforward way to the case $n_r<n$, because the image $HL$ of an infinite lattice $L$ will no longer be a lattice, and the minimum distance in $HL$ will be zero. However, when considering finite constellations $\mathcal{C}_L$, it is still possible to find suitable lower bounds on the minimum distance of the received constellation $H\mathcal{C}_L$, as will be shown in the following sections.    
\end{remark}

\section{Achievable rates for general channels} \label{general_channels_section}

Suppose that we have an infinite family of lattices  $L_k\in  \C^k$ with Hermite invariants satisfying $\frac{\mathrm{h}(L_k)}{k}\geq c$, for some positive constant $c$. Then a classical result in information theory states that with this family of lattices, all rates satisfying
$$
R < \log P + \log\left(\frac{\pi e}{4}\right)+\log c,
$$
are achievable in the additive complex  Gaussian channel \cite[Chapter 3]{CS}. This means that we can attach a single number $\mathrm{h}(L_k)$ to each lattice $L_k\in \C^k$, which roughly describes its performance and in particular estimates how close to the capacity a family of lattices can get.  
The following theorem can be seen as an analogue of this result for fading channels. 

\begin{theorem} \label{prop_positive_rate}
Suppose that $n_r \geq n$, and let $\{H_i\}_{i \in \Z}$ be a fading process such that $H_i \in M_{n_r \times n}$ is full-rank with probability $1$, and suppose that the weak law of large numbers holds for the random variables $\{\log \det (H_i^{\dagger}H_i)\}$, i.e. $\exists \mu >0$ such that $\forall \epsilon >0$,
\begin{equation} \label{WLLN}
\lim_{k \to \infty} \mathbb{P}\left\{ \abs{\frac{1}{k} \sum_{i=1}^k \log \det(H_i^{\dagger}H_i)-\mu}>\epsilon\right\}=0.
\end{equation}
Let $L_{n,k} \subset M_{n \times nk}(\C)$ be a family of $2n^2k$-dimensional multiblock lattice codes such that 
\begin{equation}\label{twocond}
\mindet{L_{n,k}}=1, \quad \text{and} \quad \Vol(L_{n,k})^{\frac{1}{n^2k}} \leq C_L
\end{equation}
for some constant $C_L > 0$. Then, any rate 
$$R<\mu + n\left(\log P - \log C_L+ \log \frac{\pi e}{4n^2}\right)$$
is achievable using the codes $L_{n,k}$ both with ML decoding and naive lattice decoding.
\end{theorem}

\begin{remark}
We note that existence of a family of lattices with 
$$
C_L \leq 23^{\frac{(n-1)}{10n}}\left(\frac{G}{2}\right),
$$
was given in Proposition \ref{prop:volume}.

\end{remark}

\begin{remark}\label{detform}
This theorem  is stated by giving two conditions   \eqref{twocond} for  the  lattices $L_{n,k}$. However, according to Lemma \ref{scale} we could have captured both of these conditions by an equivalent assumption $\delta(L_{n,k})^{2/nk}\geq \frac{1}{C_L}$. Proposition \ref{motivation} then transforms this condition to
$$
\mathrm{rh_G}(L_{n,k})= nk \left( \delta(L_{n,k})\right)^{2/nk}\geq \frac{nk}{C_L}.
$$
As only $k$ is growing, we can further write that  $\mathrm{rh_G}(L_{n,k})\geq n^2k C_L'$, where 
$C_L'=n/C_L$. We can therefore see that conditions \eqref{twocond} assure that the family of lattices $L_{n,k}$ is asymptotically good in the sense of Section \ref{statement}.

The achievable rate $R$ in  Theorem  \ref{prop_positive_rate} can then be seen as a complete analogue to the classical sphere packing result in AWGN channels.
\end{remark}

\begin{remark}
The condition (\ref{WLLN}) holds in particular for ergodic stationary fading channels and for constant MIMO channels. These special cases will be analyzed further in Section \ref{constant_gap_section}, where we will show that the codes in Theorem \ref{prop_positive_rate} achieve a constant gap to channel capacity. 
\end{remark}

To prove Theorem \ref{prop_positive_rate}, we need the following Lemma:

\begin{lemma} \label{sphere_bound}
Consider the finite code $\mathcal{C}=B(\sqrt{Pkn}) \cap (X_R + \alpha L_{n,k})$ defined in Section \ref{lattice_basics}. 
Suppose that the receiver performs maximum likelihood decoding or \vv{naive} lattice decoding (closest point search in the infinite 
lattice). Then, under the hypotheses of Theorem \ref{prop_positive_rate}, $\forall \epsilon>0$ the error probability is bounded by 
\begin{equation} \label{error_probability_multiblock}
P_e \leq 2 e^{-\frac{kn^2\epsilon^2}{8}} 
+ \mathbb{P}\bigg\{ \frac{\alpha^2}{4n}\prod_{i=1}^k {\det(H_i^{\dagger}H_i)}^{\frac{1}{nk}} <1 + \epsilon\bigg\}  
\end{equation}
\end{lemma}

\begin{IEEEproof}
We distinguish two cases: the symmetric case where $n_r=n$, and the asymmetric case with $n_r>n$.
Let $d_H$ denote the minimum Euclidean distance in the received constellation:
$$d_{H}^2=\min_{\substack{X, \bar{X} \in \mathcal{C}\\ X \neq \bar{X}}} \sum_{i=1}^k \norm{H_i (X_i-\bar{X}_i)}^2.$$
\paragraph{Case $n_r=n$} 
We have 
$$P_e \leq \mathbb{P}\left\{ \norm{W}^2 \geq \left(\frac{d_{H}}{2}\right)^2\right\},$$ 
where $W=(W_1,W_2,\ldots,W_k)$ is the multiblock noise. By the law of total probability, $\forall \epsilon>0$ we have 
\begin{equation} \label{error_probability_multiblock2}
P_e \leq \mathbb{P}\left\{ \frac{\norm{W}^2}{kn^2} \geq 1+\epsilon\right\}  + \mathbb{P} \left\{ \frac{d_H^2}{4kn^2} < 1+\epsilon\right\}.
\end{equation}
Note that $2\norm{W}^2 \sim \chi^2(2kn^2)$, and the tail of the chi-square distribution is bounded as follows for $\epsilon \in (0,1)$ \cite{Laurent_Massart}:
\begin{equation} \label{chi_square_bound}
\mathbb{P}\left\{\frac{\norm{W}^2}{kn^2} \geq 1 + \epsilon \right\} \leq 2 e^{-\frac{kn^2\epsilon^2}{8}}.
\end{equation}  
Thus, the first term in equation \eqref{error_probability_multiblock2} vanishes exponentially fast as $k \to \infty$. \\
In order to provide an upper bound for the second term, we consider a lower bound for the minimum distance in the received constellation. We have 
{\allowdisplaybreaks
\begin{align*}
&d_{H}^2 \geq \alpha^2nk \min _{X \in L_{n,k} \setminus \{0\}} \prod_{i=1}^k \abs{\det(H_iX_i)}^{\frac{2}{nk}} \geq \\
& 
\hspace{-1ex}\geq   \alpha^2 nk\prod_{i=1}^k \abs{\det(H_i)}^{\frac{2}{nk}},
\end{align*}
}%
where the  first bound comes from Lemma \ref{hadamard} and the second from
the hypothesis that $\mindet{L_{n,k}}=1$. Therefore, the second term in (\ref{error_probability_multiblock2}) is upper bounded by  
\begin{equation} 
\mathbb{P}\bigg\{ \frac{\alpha^2}{4n}\prod_{i=1}^k \abs{\det(H_i)}^{\frac{2}{nk}} <1 + \epsilon\bigg\}.  
\end{equation}
\paragraph{Case $n_r>n$}  
In this case, the lattice $H_d L_{n,k}$ is $2n^2k$-dimensional but is contained in a $2n_rnk$-dimensional space. For all $i=1, \ldots, k$, consider the QR decomposition 
$$H_i=Q_iR_i, \quad Q_i \in M_{n_r \times n_r(\C)}, \;\;R_i \in  M_{n_r \times n}(\C),$$
where $Q_i$ is unitary and $R_i$ is upper triangular. 
We have $Q_i=[ Q_i'\; Q_i'']$, where $Q_i' \in M_{n_r \times n}(\C)$ is such that $(Q_i')^{\dagger}Q_i'=I_n$, and $R_i=\left[\begin{array}{c} R_i' \\ 0 \end{array} \right]$, with $R_i' \in M_n(\C)$ upper triangular. Note that the \vv{thin} QR decomposition $H_i=Q_i'R_i'$ also holds.\\ 
Multiplying the channel equation (\ref{eq:channel}) by $Q_i^{\dagger}$, we obtain the equivalent system 
$$\tilde{Y}_i=Q_i^{\dagger} Y_i=R_i X_i+ Q_i^{\dagger}W_i$$
for all $i=1, \ldots, k$. Note that 
$$\tilde{Y}_i=\left[\begin{array}{c} Y_i' \\ Y_i'' \end{array} \right]=\left[\begin{array}{c} R_i'X_i + Q_i' W_i \\ Q_i''W_i \end{array}\right].$$
Thus, the second component contains only noise and no information.  The output of the naive lattice decoder can be written as 
\begin{align*}
&\hat{X}=\argmin_{X' \in \alpha L_{n,k}}  \sum_{i=1}^k \norm{Y_i-H_iX_i'}^2=\\
&=\argmin_{X' \in \alpha L_{n,k}}  \sum_{i=1}^k \norm{\tilde{Y}_i-R_iX_i'}^2= \\
&=\argmin_{X' \in \alpha L_{n,k}}  \sum_{i=1}^k \left(\norm{Y_i'-R_i'X_i'}^2 + \norm{Q_i'' W_i}^2\right)=\\
&=\argmin_{X' \in \alpha L_{n,k}}  \sum_{i=1}^k \norm{Y_i'-R_i'X_i'}^2, 
\end{align*}
since the second component does not depend on the lattice point $X'$. Thus, the naive lattice decoder for the original system declares an error if and only if the naive lattice decoder for the $(n,n,k)$ multiblock system with components $$Y_i'=R_i'X_i + Q_i' W_i=R_i'X_i +W_i'$$ does. Note that $W_i'=(Q_i')^{\dagger} W_i$ is an $n \times n$ matrix with i.i.d. Gaussian entries of variance $1$ per complex dimension.

Let $d_{R'}$ be the minimum distance in the $2n^2k$-dimensional constellation generated by $R'=[R_1',\ldots,R_k']$:  
$$d_{R'}=\min_{\substack{X, \bar{X} \in \mathcal{C}\\ X \neq \bar{X}}} \sum_{i=1}^k \norm{R_i' (X_i-\bar{X}_i)}^2.$$

Observe that $\forall i=1,\ldots,k$,
\begin{align*} 
&\norm{H_i (X_i-\bar{X}_i)}^2=\norm{Q_i' R_i' (X_i-\bar{X}_i)}^2=
\norm{R_i' (X_i-\bar{X}_i)}^2
\end{align*} 

Thus, $d_H=d_R$. Moreover, $\det(H_i^{\dagger}H_i)=\det((R_i')^{\dagger}R_i')=\abs{\det(R_i')}^2$. Similarly to the symmetric case, the error probability can be bounded by 
$$P_e \leq \mathbb{P}\left\{ \norm{W'}^2 \geq \left(\frac{d_{R'}}{2}\right)^2\right\},$$
where $W'=(W_1',\ldots,W_k')$. We can write
\begin{align*}
&d_{R'}^2 \geq  \alpha^2 n k \prod_{i=1}^k \abs{\det(R_i')}^{\frac{2}{nk}}= \alpha^2 n k \prod_{i=1}^k \det(H_i^{\dagger} H_i)^{\frac{1}{nk}}.
\end{align*}
The proof then follows exactly the same steps as in the symmetric case.\end{IEEEproof}

\begin{IEEEproof}[Proof of Theorem \ref{prop_positive_rate}]
The second term in (\ref{error_probability_multiblock})  can be rewritten as 
$$\mathbb{P}\left\{ \frac{1}{k} \sum_{i=1}^k \frac{1}{n} \log {\det(H_i^{\dagger}H_i)} < \log \left( \frac{4n(1+\epsilon)}{\alpha^2}\right)\right\},$$
and will vanish as long as
$$ \log \left( \frac{4n(1+\epsilon)}{\alpha^2}\right) < \frac{\mu}{n}. $$
Recalling that the normalization constant $\alpha^2$ in equation (\ref{alpha_multiblock}) satisfies 
$$\alpha^2 \geq\frac{C_{n,k}^{1/n^2k}P}{2^{R/n}C_L}$$
under the hypothesis that $\Vol(L_{n,k})^{\frac{1}{n^2k}} \leq C_L$, a sufficient condition to have vanishing error probability is
\begin{align*}
\frac{R}{n} < \log P + \frac{\mu}{n}- \log(4n(1+\epsilon)) -\log C_L + \frac{1}{n^2 k} \log C_{n,k}
\end{align*}
From Stirling's approximation, for large $k$ we have 
\begin{equation} \label{C_Stirling}
(C_{n,k})^{\frac{1}{n^2k}} \approx \pi e/(n(2\pi n^2 k)^{\frac{1}{2n^2k}}).
\end{equation}

Since $ \frac{1}{2nk} \log 2\pi n^2 k \to 0$ when $k \to \infty$, any rate 
\begin{multline*}
R < \mu + n(\log P - \log(4n(1+\epsilon)) -\log C_L +  \log \pi e - \log n)
\end{multline*}
is achievable. This holds $\forall \epsilon>0$, and concludes the proof. 
\end{IEEEproof}

\begin{remark} \label{one_sided_convergence}
Note that the two-sided convergence in probability in equation (\ref{WLLN}) is actually not required in the proof of Theorem \ref{prop_positive_rate}. The theorem still holds provided that $\forall \epsilon >0$,
\begin{equation} \label{one_sided_LLN}
\lim_{k \to \infty} \mathbb{P}\left\{ \mu - \frac{1}{k} \sum_{i=1}^k \log \det(H_i^{\dagger}H_i)>\epsilon\right\}=0.
\end{equation}
Moreover, if we have exponentially fast convergence in  (\ref{one_sided_LLN}), then the error probability $P_e$ also vanishes exponentially fast when $k \to \infty$. 
\end{remark}

As a final remark, we note that we can prove an analogue of Theorem \ref{prop_positive_rate} also in the case $n_r<n$, although the bound on achievable rates is more involved:

\begin{theorem} \label{prop_positive_rate2}
Suppose that $n_r < n$, and let $\{H_i\}_{i \in \Z}$ be a fading process such that $H_i \in M_{n_r \times n}$ is full-rank with probability $1$. Suppose that the weak law of large numbers
holds for the random variables $\log \det (H_iH_i^{\dagger})$, i.e. $\exists \mu >0$ such that $\forall \epsilon >0$,
\begin{equation} \label{WLLN2}
\lim_{k \to \infty} \mathbb{P}\left\{ \abs{\frac{1}{k} \sum_{i=1}^k \log \det(H_iH_i^{\dagger})-\mu}>\epsilon\right\}=0.
\end{equation}
Let $L_{n,k} \subset M_{n \times nk}(\C)$ be a family of $2n^2k$-dimensional multiblock lattice codes satisfying (\ref{twocond}).  Then, any rate 
$$R< \mu +n_r(\log P - 2) +(n-n_r) \log (n-n_r) +n \log \frac{\pi e}{n^2C_L}$$
is achievable using the codes $L_{n,k}$ both with ML decoding and naive lattice decoding.
\end{theorem}

The proof of Theorem \ref{prop_positive_rate2} can be found in Appendix \ref{proof_positive_rate2}.

\section{Achieving constant gap to capacity for slow fading and ergodic channels}  \label{constant_gap_section} 
\subsection{Slow fading channel} \label{slow_fading_section}
We now consider a slow fading scenario, where $H_i=H$ is constant. When $H$ is known both at the transmitter and receiver, the channel capacity is given by \cite{Tel} 
$$C(P)=\max_{Q_{\mathbf{x}}\geq 0, \tr(Q_{\mathbf{x}})\leq P} \log \det (I_{n_r} + H Q_{\mathbf{x}} H^{\dagger}),$$
where  $Q_{\mathbf{x}}$ is the covariance matrix of the input $\mathbf{x}$ for a single channel use.\\
However, if the channel is known at the receiver but not at the transmitter, the transmitter cannot use optimal power allocation and waterfilling, and can only achieve the \emph{white-input capacity} corresponding to uniform power allocation $ Q_{\mathbf{x}}=\frac{P}{n} I_n$:
$$C_{\text{WI}}= \log \det \left(I_{n_r} + \frac{P}{n} H H^{\dagger}\right)=\log \det \left(I_{n} + \frac{P}{n}  H^{\dagger}H\right).$$  
This is for example the case for an open-loop broadcast channel where the transmitter cannot perform rate adaptation for all the users.\\
Clearly, Theorems \ref{prop_positive_rate} and \ref{prop_positive_rate2} apply to the slow-fading scenario since the law of large numbers holds. Moreover, the convergence of the error probability to zero will be exponential, since the second term in equation \eqref{error_probability_multiblock} is actually zero.  The following corollary then shows that a constant gap to white-input capacity is achievable:
\begin{cor} \label{prop_slow_fading}
Consider a slow fading channel such that $H_i=H$ for all $i \geq 1$, and let $L_{n,k} \subset M_{n \times nk}(\C)$ be a family of $2n^2k$-dimensional multiblock lattice codes such that 
$\mindet{L_{n,k}}=1$ and $\Vol(L_{n,k})^{\frac{1}{n^2k}} \leq C_L$. 
Then, this coding scheme can achieve any rate
$$R< \log \det \frac{P}{n} H^{\dagger}H-n\log C_L + n\log \frac{\pi e}{4n}$$
if $n_r \geq n$, and any rate 
$$R< \log \det \frac{P}{n} HH^{\dagger}- 2n_r-(n-n_r) \log \frac{n}{n-n_r}+ n \log \frac{\pi e}{n C_L} $$
if $n_r<n$.
\end{cor} 

\begin{remark} \label{remark_constant_gap}
In the case $n_r \geq n$, let $\lambda_i$, $i=1,\ldots,n$ be the singular values of $H$. Then the channel capacity can be written as  
$C(P)=\sum_{i=1}^n \log \left(1+\frac{P}{n}\lambda_i\right).$
The previous corollary shows that the achievable rate is of the form 
$$R(P)=\max\left(0,\log\det\frac{P}{n} H^{\dagger}H -c\right)$$
for some constant $c>0$. Let $P_{\min}$ be the smallest value of $P$ such that $R(P)>0$ if $P>P_{\min}$. Then, for $P \leq P_{\min}$ we have that $C(P)-R(P)=C(P) \leq C(P_{\min})$, while for $P>P_{\min}$,
$C(P)-R(P)=\sum_{i=1}^n \log \left(1+\frac{n}{P\lambda_i}\right)+c$ which is a strictly decreasing function of $P$ and tends to $c$ when $P \to \infty$. Thus, for all $P>0$ we have that $C(P)-R(P)\leq C(P_{\min})$. This shows that the gap is bounded by a constant, however the value of $P_{\min}$ and of the constant depends on the channel. \\
A similar argument holds for $n_r <n$. 
\end{remark}

\subsection{Stationary ergodic channels} \label{ergodic_section}
We now specialize the results of Section \ref{general_channels_section} to the case where the fading process $\{H_i\}$ is \emph{ergodic} and \emph{stationary}. For the sake of completeness, we review the relevant definitions here.

Let $\mathcal{\I}$ be the set $\Z$ or $\mathbb{N}$, and consider a random process $X^{\mathcal{I}}=\{X_i\}_{i \in \mathcal{I}}$ on a probability space $(\Omega,\mathcal{B},\mathbb{P})$ where each random variable $X_i$ takes values in a separable Banach space $\mathcal{X}$. The sequence space $(\mathcal{X}^{\mathcal{I}},\mathcal{B}(\mathcal{X}^{\mathcal{I}}))$  with the Borel sigma-algebra inherits a probability measure $m_{X}$ from the underlying probability space, defined by 
\begin{equation} \label{m_X}
m_{X}(A)=\mathbb{P}\left\{ \omega: \; X^{\mathcal{I}}(\omega) \in A\right\} \quad \forall A \in \mathcal{B}(\mathcal{X}^{\mathcal{I}}).
\end{equation}
\begin{defn}
The process $\{X_i\}$ is called \emph{stationary} if $\forall t, k \in \mathbb{N}, \forall i_1,i_2,\ldots,i_k \in \mathcal{I}$,  the joint distribution of $(X_{i_1},X_{i_2},\ldots,X_{i_k})$ is the same as that of $(X_{i_1+t},X_{i_2+t},\ldots,X_{i_k+t})$. 
\end{defn}

In this case it is well-known \cite[p. 494]{Billingsley} that the measure $m_{X}$ is invariant with respect to the shift map $T: \mathcal{X}^{\mathcal{I}} \to \mathcal{X}^{\mathcal{I}}$ such that $T(\{x_i\})=\{x_{i+1}\}$. 
\begin{defn}
The process $\{X_i\}$ is called \emph{ergodic} if $\forall A \in \mathcal{B}(\mathcal{X}^{\mathcal{I}})$ such that $T^{-1}(A)=A$, we have that $m_X(A)$ is equal to $0$ or $1$. 
\end{defn}

We now go back to the channel model (\ref{eq:channel}). For the sake of simplicity, we suppose that $n_r \geq n$. If the fading process $\{H_i\}$ is stationary and ergodic, it is not hard to see that the random process $\{X_i\}=\left\{ \log \det(H_i^{\dagger}H_i)\right\}$ taking values in $\mathcal{X}=\R$ is also stationary and ergodic, and the shift $T: \R^{\mathcal{I}} \to \R^{\mathcal{I}}$ preserves the measure $m_{X}$ defined in (\ref{m_X}).\\
For an ergodic process such that the shift $T$ is measure-preserving, Birkhoff's theorem \cite{Pollicott_Yuri} guarantees that for any $f \in L^1(\mathcal{X}^{\mathcal{I}},\mathcal{B}(\mathcal{X}^{\mathcal{I}}),m_X)$, the sample means with respect to $f$ converge almost everywhere: for almost all $\{x_i\} \in \mathcal{X}^{\mathcal{I}}$,
\begin{equation} \label{Birkhoff}
\lim_{k \to \infty} \frac{1}{k} \sum_{n=1}^{k} f(T^n(\{x_i\}))= \int_{\mathcal{X}^{\mathcal{I}}} f dm_X.
\end{equation}
In particular, the projection $\Pi: \R^{\mathcal{I}} \to \R$ on the first coordinate is $L^1$ according to the image measure $m_{X}$ if and only if 
$\mathbb{E}\left[\abs{\log \det H^\dagger H}\right]< \infty$.
Under this hypothesis, Birkhoff's theorem implies the law of large numbers:
\begin{align} \label{LLN}
&\lim_{k \to \infty} \frac{1}{k} \sum_{i=1}^{k} X_i= \int_{\R^{\mathcal{I}}} \Pi(\{x_i\}) dm_X(\{x_i\})= \notag\\
&=\int_{\Omega} \Pi \circ X^{\mathcal{I}} d \mathbb{P}=\int_{\Omega} X_1 d\mathbb{P}=\mathbb{E}[X] \quad \text{a.e.}
\end{align}
In other words,
\begin{equation} \label{SLLN} 
\lim_{k \to \infty}  \frac{1}{k} \sum_{i=1}^k \log \det(H_i^{\dagger}H_i) = \mathbb{E}_H\left[ \log \det (H^{\dagger}H)\right]
\end{equation}
almost everywhere.\\ 
In the ergodic stationary case, it is well-known \cite{Tel, Caire_Elia_Kumar} that the ergodic capacity of the channel is well-defined and does not depend on the channel correlation with respect to time, but only on its first order statistics. Given a power constraint $P$ in equation (\ref{power_constraint}), the ergodic capacity (per channel use) is equal to
$$C(P)=\max_{Q_{\mathbf{x}}\geq 0, \tr(Q_{\mathbf{x}})\leq P} \mathbb{E}_{H}\left[\log \det (I_{n_r} + H Q_{\mathbf{x}} H^{\dagger})\right],$$
where $H$ is a random matrix with the same first-order distribution of the process $\{H_i\}$, which is independent of time by stationarity, and $Q_{\mathbf{x}}$ is the covariance matrix of the input $\mathbf{x}$ for one channel use\footnote{We note that the capacity (per channel use) of the block fading MIMO channel of finite block length $T$ with perfect channel state information at the receiver is independent of $T$ \cite[eq. (9)]{Marzetta_Hochwald}. So the previous result still holds in the multiblock case.}.\\
If we suppose that the channel is \emph{isotropically invariant}, i.e. the distribution of $H$ is invariant under right multiplication by unitary matrices, then the optimal input covariance matrix is $Q_{\mathbf{x}}=\frac{P}{n}I_n$ \cite{Tel} and we have 
$$C(P)=\mathbb{E}_{H}\left[\log \det \Big(I_{n_r} + \frac{P}{n} H H^{\dagger}\Big)\right].$$
Since $\det(I_{n_r}+\frac{P}{n}HH^{\dagger})=\det(I_n+\frac{P}{n}H^{\dagger}H)$, we can also write 
$$C(P)=\mathbb{E}_{H}\left[\log \det \Big(I_{n} + \frac{P}{n} H^{\dagger}H\Big)\right].$$
The following Corollary to Theorem \ref{prop_positive_rate} shows that in this case, the proposed multiblock codes can achieve a constant gap to ergodic capacity. 

\begin{cor} \label{corollary_ergodic}
Suppose that $n_r \geq n$ and that the fading process $\{H_i\}$ is ergodic, stationary and isotropically invariant. Moreover, suppose that $\mathbb{E}\left[\abs{\log \det H^\dagger H}\right]< \infty$. Let $L_{n,k} \subset M_{n \times nk}(\C)$ be a family of $2n^2k$-dimensional multiblock lattice codes such that 
$\mindet{L_{n,k}}=1$ and $\Vol(L_{n,k})^{\frac{1}{n^2k}} \leq C_L$. Then, any rate
$$R<\mathbb{E}_H\left[\log \det \frac{P}{n} H^{\dagger} H\right] - n \log C_L +n\log \frac{\pi e}{4n}$$
is achievable using the codes $L_{n,k}$ both with ML decoding and naive lattice decoding.
\end{cor}
\begin{IEEEproof}
From equation (\ref{SLLN}), we have that the hypotheses of Theorem \ref{prop_positive_rate} are satisfied (actually, only the weak law of large numbers was required). Consequently, any rate 
\begin{align*}
&R< n\log P +\mathbb{E}_H\left[\log \det H^{\dagger} H\right] - n \log C_L +n\log \frac{\pi e}{4n^2}=\\
&=\log\left(\frac{P}{n}\right)^n + \mathbb{E}_H\left[\log \det H^{\dagger} H\right] - n \log C_L +n\log \frac{\pi e}{4n}=\\
&=\mathbb{E}_H\left[\log \det \frac{P}{n} H^{\dagger} H\right] - n \log C_L +n\log \frac{\pi e}{4n}
\end{align*} 
is achievable. 
\end{IEEEproof}
A similar corollary to Theorem \ref{prop_positive_rate2} holds in the case $n_r <n$: 
\begin{cor}
Suppose that $n_r < n$ and that the fading process $\{H_i\}$ is ergodic, stationary and isotropically invariant. Moreover, suppose that $\mathbb{E}\left[\abs{\log \det H H^\dagger}\right]< \infty$. Let $L_{n,k} \subset M_{n \times nk}(\C)$ be a family of $2n^2k$-dimensional multiblock lattice codes such that 
$\mindet{L_{n,k}}=1$ and $\Vol(L_{n,k})^{\frac{1}{n^2k}} \leq C_L$. Then, any rate $R$ lower than
$$\mathbb{E}_H\Big[\log \det \frac{P}{n} H H^{\dagger}\Big] - 2n_r-(n-n_r) \log \frac{n}{n-n_r}+ n \log \frac{\pi e}{n C_L}$$
is achievable using the codes $L_{n,k}$ both with ML decoding and naive lattice decoding.
\end{cor}

\begin{remark} \label{remark_constant_gap2}
Using the same argument as in Remark \ref{remark_constant_gap}, we can show that the achievable rate is within a constant gap from capacity, although this constant will depend on the fading model. \\
Under the hypothesis $\mathbb{E}_H\left[\abs{\log\det H^{\dagger}H}\right]$, we have $\abs{\mathbb{E}_H\left[\log \det H^{\dagger}H\right]}<\infty$. The achievable rate is of the form 
$R(P)=\max\left(0,n\log\frac{P}{n}+\mathbb{E}_H[\log\det H^{\dagger}H] -c\right)$
for some constant $c>0$. Let $P_{\min}$ be the smallest value of $P$ such that $R(P)>0$ if $P>P_{\min}$. For $P \leq P_{\min}$ we have that $C(P)-R(P)=C(P) \leq C(P_{\min})$. \\
Let $\lambda_i$, $i=1,\ldots,n$ be the (random) singular values of $H$. For $P>P_{\min}$,
$C(P)-R(P)=\sum_{i=1}^n \mathbb{E}_H\left[\log \left(1+\frac{n}{P\lambda_i}\right)\right]+c$ which is a strictly decreasing function of $P$ and tends to $c$ when $P \to \infty$. This shows that the gap is uniformly bounded. 
\end{remark}

\section{Achievable rates and error probability bounds for i.i.d. Rayleigh fading channels} \label{Rayleigh_fading}

We now suppose that the coefficients of $H_i$ are circular symmetric complex Gaussian with zero mean and unit variance per complex dimension, and that the fading blocks $H_i$ are independent. In this case, the achievable rate can be computed explicitly, and we can prove that the error probability vanishes exponentially fast. \\
Let $\psi(x)=\frac{d}{dx} \ln \Gamma(x)$ denote the Digamma function. Then we have the following:
\begin{prop} \label{prop_Rayleigh}
Let $L_{n,k} \subset M_{n \times nk}(\C)$ be a family of $2n^2k$-dimensional multiblock lattice codes such that 
$\mindet{L_{n,k}}=1$ and $\Vol(L_{n,k})^{\frac{1}{n^2k}} \leq C_L$. Then, over the $(n,n_r,k)$ multiblock channel, these codes achieve any rate
\begin{equation*}
R < \mathbb{E}_H\left[\log \det \frac{P}{n} H^{\dagger}H\right]  - n\log C_L + n \log\frac{\pi e}{4n},
\end{equation*}
where 
\begin{equation} \label{rate_Gaussian}
\mathbb{E}_H\left[\log \det \frac{P}{n} H^{\dagger}H\right]= \log\frac{P}{n}  e^{\frac{1}{n}\sum\limits_{i=n_r-n+1}^{n_r} \psi(i)}.
\end{equation}
Moreover, the error probability vanishes exponentially fast. 
\end{prop}


\begin{IEEEproof}[Proof of Proposition \ref{prop_Rayleigh}]
The first statement follows from Corollary \ref{corollary_ergodic}. The next step is to prove equation (\ref{rate_Gaussian}). It is well-known \cite{Goodman,Edelman} that if $H$ is an $n_r \times n$ matrix with i.i.d. complex Gaussian entries having variance per real dimension $1/2$, the random variable $\det(H^{\dagger}H)$, corresponding to the determinant of the Wishart matrix $H^{\dagger}H$, is distributed as the product 
$$V_{n,n_r}= Z_{n_r-n+1} Z_{n_r-n+2} \cdots  Z_{n_r}$$ 
of $n$ independent variables, such that $\forall j= n_r-n+1,\ldots,n_r$, $2Z_j$ is a chi square random variable with $2j$ degrees of freedom. The density of $Z_j$ is $p_{Z_j}(x)=\frac{x^{j-1}e^{-x}}{\Gamma(j)}$.
We have
\begin{align*}
&\mathbb{E}[\ln Z_j]= \frac{1}{\Gamma(j)} \int_0^{\infty}  x^{j-1} e^{-x} \ln x\; dx= \psi(j),\\
&M_{n,n_r}=\mathbb{E}[\ln V_{n,n_r}]=\sum_{j=n_r-n+1}^{n_r} \psi(j)=\mathbb{E}_H\left[\ln \det H^{\dagger} H\right].
\end{align*}
Then if we consider the base $2$ logarithm, we find 
\begin{align*}
&\mathbb{E}_{H}\left[\log \det H^{\dagger} H\right] =\mathbb{E}[\log V_{n,n_r}]=\frac{M_{n,n_r}}{\ln2}=\\
&=\frac{\sum_{j=n_r-n+1}^{n_r} \psi(j)}{\ln2}=n \log e^{\frac{1}{n} \sum_{j=n_r-n+1}^{n_r} \psi(j)}
\end{align*}
which concludes the proof of equation (\ref{rate_Gaussian}).\\
In order to show that the error probability converges exponentially fast, by Remark \ref{one_sided_convergence} it is enough to show that we have exponential convergence in equation (\ref{one_sided_LLN}). \\  
Consider a sequence of i.i.d. random variables $\ln V_{n,n_r}^{(i)}$, $i=1,\ldots,k$, with the same distribution as $\ln V_{n,n_r}$. Using the Chernoff bound \cite{Proakis}, given $\delta>0$, $\forall v>0$ we have
\begin{align}
&\mathbb{P}\left\{ \frac{M_{n,n_r}}{\ln 2} -\frac{1}{k} \sum_{i=1}^k \log \det H_i^{\dagger}H_i \geq \frac{\delta}{\ln 2}\right\}= \notag\\
&=\mathbb{P}\left\{ M_{n,n_r} -\frac{1}{k} \sum_{i=1}^k \ln \det H_i^{\dagger}H_i \geq \delta\right\}= \notag\\
&=\mathbb{P}\left\{ M_{n,n_r} -\frac{1}{k} \sum_{i=1}^k \ln V_{n,n_r}^{(i)} \geq  \delta\right\} \leq \notag\\
&\leq  e^{kv(M_{n,n_r}-\delta)}\left(\mathbb{E}[e^{-v\ln V_{n,n_r}}]\right)^k \label{Chernoff_bound}
\end{align}
The tightest bound in (\ref{Chernoff_bound}) is obtained for $v_{\delta}$ such that 
\begin{equation*} 
\mathbb{E}[-\ln V_{n,n_r} e^{-v_{\delta} \ln V_{n,n_r}}]= (\delta-M_{n,n_r}) \mathbb{E}[e^{-v_{\delta} \ln V_{n,n_r}}].
\end{equation*}
Observe that 
{\allowdisplaybreaks
\begin{align}
&\mathbb{E}[Z_j^{-v}]=\frac{1}{\Gamma(j)}  \int_0^{\infty} x^{j-1-v} e^{-x} dx=\frac{\Gamma(j-v)}{\Gamma(j)}, \label{eq1}\\
& \mathbb{E}[Z_j^{-v}\ln Z_j]=\frac{1}{\Gamma(j)} \int_0^{\infty} x^{j-1-v} e^{-x} \ln x \; dx = \notag \\
&=\frac{\Gamma(j-v)}{\Gamma(j)}\psi(j-v). \label{eq2}
\end{align} }
Thus we find
{\allowdisplaybreaks
\begin{align*}
&\mathbb{E}\left[e^{-v\ln V_{n,n_r}}\right]=\mathbb{E}\left[V_{n,n_r}^{-v}\right]=\prod_{j=n_r-n+1}^{n_r} \mathbb{E}\left[Z_j^{-v}\right]=\\
&=\prod_{j=n_r-n+1}^{n_r} \frac{\Gamma(j-v)}{\Gamma(j)},\\
&\mathbb{E}\left[-\ln V_{n,n_r} e^{-v \ln V_{n,n_r}}\right]=\mathbb{E}\left[-V_{n,n_r}^{-v} \ln V_{n,n_r}\right]=\\
&=\sum_{j=n_r-n+1}^{n_r} \mathbb{E}\left[-\ln Z_j \prod_{l=n_r-n+1}^{n_r} Z_l^{-v}\right]=\\
&=\sum_{j=n_r-n+1}^{n_r} \left(\prod_{l\neq j} \mathbb{E}[Z_l^{-v}]\right) \mathbb{E}[-Z_j^{-v}\ln Z_j]=\\
&=-\sum_{j=n_r-n+1}^{n_r} \left(\prod_{l\neq j} \frac{\Gamma(l-v)}{\Gamma(l)}\right) \frac{\Gamma(j-v)}{\Gamma(j)} \psi(j-v)=\\
&=-\prod_{l=n_r-n+1}^{n_r}\frac{\Gamma(l-v)}{\Gamma(l)} \sum_{j=n_r-n+1}^{n_r} \psi(j-v)   
\end{align*}
}
Consequently, the tightest bound in (\ref{Chernoff_bound}) is achieved for $v_\delta$ such that
\begin{align} \label{v_delta}
&\delta=\sum_{l=n_r-n+1}^{n_r} (\psi(l)-\psi(l-v_{\delta})).
\end{align}
Note that as $\delta \to 0$, $v_{\delta} \to 0$. The right-hand side in equation (\ref{Chernoff_bound}) for $v=v_{\delta}$ can be rewritten as 
{\allowdisplaybreaks
\begin{align*}
&e^{kv_{\delta}(-\delta+\sum_{j=n_r-n+1}^{n_r} \psi(j))}\left(\prod_{l=n_r-n+1}^{n_r} \frac{\Gamma(l-v_{\delta})}{\Gamma(l)}\right)^k =\\
&=e^{k(-v_{\delta} \delta +\sum_{j=n_r-n+1}^{n_r}(v_{\delta} \psi(j) +\ln \Gamma(j-v_{\delta})-\ln\Gamma(j)))}=\\
&=e^{k\sum_{j=n_r-n+1}^{n_r}(v_{\delta}\psi(j-v_{\delta})-\ln \Gamma(j)+\ln\Gamma(j-v_{\delta}))}
\end{align*}
}%
using (\ref{v_delta}).\\
Recall that $\Gamma(x)$ is monotone decreasing for $0 <x < a_0=1.461632\ldots$ and monotone increasing for $x>a_0$. Using the mean value theorem for the function $\ln \Gamma(x)$ in the interval $[i-v_{\delta},i]$ we get that for $i=1$, $v_{\delta} \psi(1-v_{\delta})+\ln \Gamma(1-v_{\delta}) \leq 0$, and for $i\geq 2$, $v_{\delta} \psi(i-v_{\delta}) \leq \ln \Gamma(i) - \ln \Gamma (i-v_{\delta})$. Thus, the exponent is negative both for $n=n_r$ and for $n>n_r$. 
We can conclude that 
\begin{align*}
\mathbb{P}\left\{ \frac{M_{n,n_r}}{\ln 2} -\frac{1}{k} \sum_{i=1}^k \log \det H_i^{\dagger}H_i \geq \frac{\delta}{\ln 2}\right\}\leq e^{-k K_{n,n_r,\delta}}
\end{align*}
for some positive constant $K_{n,n_r,\delta}$. 
Since the bound holds $\forall \delta>0$, using Remark \ref{one_sided_convergence} with $\mu=\frac{M_{n,n_r}}{\ln 2}$, we find that the error probability tends to $0$ exponentially fast for any rate
\begin{align} 
R<n\left(\log \frac{P}{n} e^{\frac{1}{n} \sum_{i=n_r-n+1}^{n_r} \psi(i)}-\log C_L + \log \frac{\pi e}{4n}\right).
\tag*{\IEEEQED}
\end{align}
\let\IEEEQED\relax%
\end{IEEEproof}

\begin{cor} \label{prop_highdegree}
Over the $(n,n,k)$ multiblock channel, reliable communication is guaranteed when $k \to \infty$ for rates
\begin{equation*}
R < n\left(\log\frac{P}{n}  e^{\frac{1}{n}\sum_{i=1}^n \psi(i)}  + \log\frac{\pi e}{2n}  - \log {23}^{\frac{1}{10}\left(1-\frac{1}{n}\right)} G \right)
\end{equation*}
when using the multiblock code construction in Proposition \ref{prop:volume}. 
\end{cor}


\section{Existence of asymptotically good lattices}\label{construction}
All of our capacity results depend on the existence of lattices with asymptotically good normalized minimum determinants, which was claimed in Section \ref{statement}.
In this section we will prove this result. \\
We will first recall the construction of single-block space-time codes from cyclic division algebras (see for example \cite{OBV}). Due to space constraints, we refer the reader to \cite{Reiner} for algebraic definitions. 
\begin{definition}\label{cyclic}
Let $K$ be an algebraic number field of degree $m$ and assume   that $E/K$ is a cyclic Galois
 extension of degree $n$ with Galois group
$Gal(E/K)=\left\langle \sigma\right\rangle$. We can define an associative $K$-algebra
$$
\mathcal{A}=(E/K,\sigma,\gamma)=E\oplus uE\oplus u^2E\oplus\cdots\oplus u^{n-1}E,
$$
where   $u\in\mathcal{A}$ is an auxiliary
generating element subject to the relations
$xu=u\sigma(x)$ for all $x\in E$ and $u^n=\gamma\in K^*$. We call the resulting algebra a \emph{cyclic algebra}.
\end{definition}
Here $K$ is the center of the algebra $\A$.

 \begin{definition}
 We  call $\sqrt{[\A:K]}$ the \emph{degree} of the algebra $\A$. It is easily verified that the degree of $\A$ is equal to $n$.
\end{definition}

We consider $\A$ as a right  vector space over $E$
and note that every  element $a=x_0+ux_1+\cdots+u^{n-1}x_{n-1}\in\mathcal{A}$
has the following representation as a matrix: 
\[
\phi(a)=\begin{pmatrix}
x_0& \gamma\sigma(x_{n-1})& \gamma\sigma^2(x_{n-2})&\cdots &
\gamma\sigma^{n-1}(x_1)\\
x_1&\sigma(x_0)&\gamma\sigma^2(x_{n-1})& &\gamma\sigma^{n-1}(x_2)\\
x_2& \sigma(x_1)&\sigma^2(x_0)& &\gamma\sigma^{n-1}(x_3)\\
\vdots& & & & \vdots\\
x_{n-1}& \sigma(x_{n-2})&\sigma^2(x_{n-3})&\cdots&\sigma^{n-1}(x_0)\\
\end{pmatrix}
\]

The mapping $\phi$ is called the \emph{left regular representation} of $\A$ and allows us to embed any cyclic algebra into $M_n(\C)$. Under such an embedding $\phi(\A)$ forms an $mn^2$-dimensional $\Q$-vector space. 

We are particularly interested in algebras $\A$ for which $\phi(a)$ is invertible for all non-zero $a\in\A$.

\begin{definition}\label{divisionalgebra}
A cyclic $K$-algebra $\D$ is a \emph{division algebra} if every
non-zero element of $\D$ is invertible.
\end{definition}

If we assume that $\D$ is a division algebra, then $\phi$ is an injective mapping to $M_n(\C)$ and every non-zero element in $\phi(\D)$ is invertible.
  However, $\phi(\D)$ is not a lattice. Therefore we will instead consider a  suitable subset of $\D$.

\begin{definition}
 A \emph{$\Z$-order} $\Lambda$ in $\D$ is a subring of $\D$ having the same identity element as
$\D$, and such that $\Lambda$ is a finitely generated
module over $\Z$ which generates $\mathcal{D}$ as a linear space over $\Q$.
\end{definition}

With the previous definition, the set $\phi(\Lambda)$ is a matrix lattice that can be used for coding over a single space-time block. 

A generalization of the embedding $\phi$ to the multiblock case was proposed in \cite{YB07,Lu} for division algebras whose center $K$ contains an imaginary quadratic field. In this paper we consider a more general multiblock construction developed in \cite{LSV}, which applies to any totally complex center $K$. \\
We say that a degree $2k$ number field  $K$ is totally complex if   for every $\Q$-embedding $\beta_i: K\hookrightarrow  \C$ the image $\sigma_i(K)$ includes complex elements.   The field $K$ has  $2k$ distinct $\Q$-embeddings $\beta_i: K\hookrightarrow\C$. As we assumed that $K$ is totally complex, each of these embeddings is part of a complex conjugate pair. We will denote by $\bbar{\beta_i}$ the embedding given by  $x\mapsto \bbar{\beta_i(x)}$.

For each $\beta_i$ we can find an embedding $\alpha_i: E\hookrightarrow \C$ such that $\alpha_i|_{K}=\beta_i$. This choice can be made in such a way that
$\bbar{\alpha_i}|_{K}=\bbar{\beta_i}$. 
We will suppose that the embeddings  $\{\alpha_1,\dots, \alpha_{2k}\}$  have been ordered in such a way that $\alpha_i=\bbar{\alpha_{i+m}}$, for $0\leq i\leq m$.  
Let $a$ be an element of $\D$ and $A=\phi(a)$. Consider the mapping $\varphi: \A\mapsto M_{n\times nk}(\C)$ given by
\begin{equation}\label{main_map}
a\mapsto (\alpha_1(A),\dots, \alpha_{k}(A)),
\end{equation}
where each $\alpha_i$  is extended to an embedding $\alpha_i: M_n(E)\hookrightarrow M_n(\C)$.\\
The following result was proven in \cite[Proposition 5]{LSV}:

\begin{proposition}\label{reg2}
Let $\Lambda$ be a $\Z$-order in $\D$ and $\varphi$ the previously defined embedding.
Then $\varphi(\Lambda)$ is a $2kn^2$-dimensional lattice in $M_{n \times nk}(\C)$ which satisfies
$$
\mindet{\varphi(\Lambda)}= 1, \,\, \Vol(\varphi(\Lambda))= 2^{-kn^2}\sqrt{|d(\Lambda/\Z)|}
$$
and
$$
\delta(\varphi(\Lambda))=\left(\frac{2^{2kn^2}}{|d(\Lambda/\Z)|}\right)^{1/4n}.
$$
\end{proposition}

Here  $d(\Lambda/\Z)$  is the $\Z$-discriminant of the order $\Lambda$. It is a non-zero integer we can associate to any $\Z$-order of  $\D$. We refer the reader to \cite{Reiner} for the relevant definitions. 

We can now see that in order to maximize the minimum determinant of a  multi-block code, we have to minimize the $\Z$-discriminant of the corresponding $\Z$-order $\Lambda$. 

The first step to attack this question is to assume that $\Lambda$ has some extra structure.
Let $\OO_K$ be the ring of algebraic integers of $K$. If we assume that $\Lambda$ is also an $\OO_K$ module, then the $\OO_K$-discriminant of $\Lambda$ is well-defined \cite{Reiner}, and will be denoted by $d(\Lambda/\mathcal O_K)$. The following formula holds: 
\begin{equation}
d(\Lambda/\mathbb Z)=N_{K/\mathbb Q}(d(\Lambda/\mathcal \OO_K))(d_K)^{n^2}, \label{d_Lambda_Z}
\end{equation}
where 
$N_{K/\mathbb Q}$ is the algebraic norm in $K$ and $d_K$ is the discriminant of the field $K$.

In the case of fixed center $K$, \cite{VHLR} addressed the problem of finding the division algebras with the smallest $\OO_K$-discriminant, yielding the densest MIMO lattices. The main construction is based on the following result (Theorem 6.14 in \cite{VHLR}):
\begin{theorem} \label{theorem_VHLR}
Let $K$  be a number field of degree $2k$ and $P_1$ and $P_2$ be two prime ideals of $K$. 
Then there exists a degree $n$ division algebra $\D$ having an $\OO_K$-order $\Lambda$ with discriminant
\begin{equation}\label{discriminant2}
d(\Lambda/\ZZ)=(N_{K/\mathbb Q}(P_1) N_{K/\mathbb Q}(P_2))^{n(n-1)} (d_K)^{n^2}.
\end{equation}
\end{theorem}
Theorem \ref{theorem_VHLR} suggests that in order to build families of $(n,n,k)$ multiblock codes with the largest normalized minimum determinant, we should proceed in two steps:
\begin{enumerate}
\item[a)] choose a sequence of center fields $K$ of degree $2k$ such that their discriminants $d_{K}$ grow as slowly as possible; 
\item[b)] given the center $K$, choose an algebra $\D$ satisfying (\ref{discriminant2}), where $P_1$ and $P_2$ are the prime ideals in $K$ with the smallest norms\footnote{However, we note \cite{LSV} that \emph{a priori} there may be a trade-off between these two choices, so that minimizing the two terms in (\ref{d_Lambda_Z}) separately may be suboptimal.}. 
\end{enumerate}

We now discuss the choice of a suitable sequence of center fields. The following theorem by Martinet \cite{Martinet} proves the existence of infinite sequences of totally complex number fields $K$ with small discriminants $d_K$. As we will see in the following, choosing such a field as the center of the algebra $\mathcal{D}$ is a key element to obtain a good normalized minimum determinant.   

\begin{theorem}[Martinet] \label{Martinet_theorem}
There exists an infinite tower of totally complex number fields $\{K_k\}$ of degree $2k$, where  $2k=5\cdot2^{2+t}$, such that
\begin{equation} \label{G}
 \abs{d_{K_k}}^{\frac{1}{2k}}=G,
\end{equation}
for $G \approx 92.368$.
\end{theorem}


The following Lemma shows that the number fields in the Martinet family have suitable primes of small norm yielding a good bound in Theorem \ref{theorem_VHLR}. 

\begin{lemma}\label{norm}
Every number field $K_k$ in the Martinet family has ideals $P_1$ and $P_2$ such that
$$
N_{K/\mathbb Q}(P_1)\leq 23^{k/10} \,\,\mathrm{and}\,\, N_{K/\mathbb Q}(P_2) \leq 23^{k/10}.
$$
\end{lemma}
\begin{IEEEproof}
Every field $K_k$ has a  subfield $F=\Q(\cos(2\pi/11), \sqrt{2},\sqrt{-23}) $, where $[F:\Q]=20$ (see for example \cite[p. 395]{TsVl}). 
The field $F$ has prime ideals $B_1$ and $B_2$  such that $N_{F/\Q}(B_i)=23$. Let us now suppose that $P_1$ and $P_2$ are such prime ideals of  $K_k$ that $P_i\cap \OO_F=B_i$.
  Transitivity of the norm then gives us that
\begin{align}
N_{K_k/\Q}(P_i)=N_{F/\Q}( N_{K_k/F}(P_i))\leq N_{F/\Q} (B_i)^{2k/20}. \tag*{\IEEEQED}%
\end{align}
\let\IEEEQED\relax%
\end{IEEEproof}
Armed with this observation, we can finally prove Proposition \ref{prop:volume}.


\begin{IEEEproof}[Proof of Proposition \ref{prop:volume}]

Suppose that we have a degree $2k$ field extension $K$ in the Martinet family of totally complex fields such that (\ref{G}) holds. We know that this field $K$ has some primes $P_1$ and $P_2$  such that $N_{K/\mathbb Q}(P_1)\leq 23^{k/10} \,\,\mathrm{and}\,\,N_{K/\mathbb Q}(P_2) \leq 23^{k/10}$. Then, there exists a central division algebra $\mathcal{D}$ of degree $n$ over $K$, and a maximal order $\Lambda$ of $\mathcal{D}$, such that
\begin{align}
&d(\Lambda/\mathbb{Z})=(N_{k/\mathbb Q}(P_1) N_{K/\mathbb Q}(P_2))^{n(n-1)} (d_K)^{n^2} \leq \notag \\
& \leq (23^{k/5})^{(n(n-1))} (G^{2k})^{n^2}. \tag*{\IEEEQED}%
\end{align}%
\let\IEEEQED\relax%
\end{IEEEproof}%

\begin{remark}
We note that the number field towers in Theorem \ref{Martinet_theorem} are not the best known possible.  It was shown in \cite{Hajir_Maire} that one can construct a family of 
totally complex fields such that   $G<82.2$, but this choice would add some notational complications.
\end{remark}

\section{Corollaries to single antenna fading channel}\label{numberfields}
The single antenna fast fading channel is one of the special cases of the general channel model \eqref{eq:channel}. It is particularly illuminating as the connection to the classical AWGN  lattice coding is most striking. In this case the abstract matrix lattices  of Section \ref{construction} correspond to simple number field codes that have been studied for twenty years \cite{OV}.  Due to the familiarity and simplicity of this model we can most easily compare our work to previous research on the topic.


In the single antenna  case the channel model \eqref{eq:channel} gets simplified to 
\begin{equation}\label{single_channel}
y_i=h_i \cdot x_i + w_i,
\end{equation}
where $x_i$ are the transmitted symbols, and $\forall i=1,\ldots, k$,  $w_i$ are i.i.d. complex Gaussian random variables with variance $\sigma_h^2=\sigma^2=\frac{1}{2}$ per real dimension and $\{h_i\}$ is some complex fading process such that $\sum_{i=1}^k \frac{1}{k} \log|h_i|^2$ converges in probability to some constant when the number of blocks $k$ tends to infinity.

This scenario has received considerable interest in the case of an i.i.d. complex Gaussian fading  process $\{h_i\}$, and several works have focused on the design of lattice codes for this model \cite{GB, BERB}. The analysis of the union bound for the pairwise error probability for a lattice code $L \subset \C^k$ leads to a design criterion based on the maximization of the \emph{normalized product distance} 
$$\mathrm{Nd}_{p, min}(L)=\inf_{\mathbf{x} \in L \setminus \{0\}} \frac{\prod_{i=1}^k \abs{x_i}}{\Vol(L)^{\frac{1}{2}}}.$$
Note that the normalized product distance is a special case (for $n=1$) of the normalized minimum determinant in (\ref{normalized_minimum_determinant}).
Most of the works in the literature have focused on the optimization of the product distance for lattice signal constellations with a fixed number of blocks $k$; 
few authors  
\cite{Xing, FOV}
have also studied the upper and lower bounds for  $\mathrm{Nd}_{p, min}$  over all lattices when $k$ grows to infinity.

However, there has been no general consensus on whether significant gain could be achieved
from coding over an extensive number of fading realizations. For example the authors  in \cite{BPS} state that:  ``increasing the diversity does not necessarily increase to the same extent the performance: in fact, the minimum product distance decreases and the product kissing number increases. Simulations show that most of the gain is obtained for diversity orders up to  16''.  
In fact, the analysis of the distribution of pairwise errors in the union bound 
 as in \cite{VLL2013} shows that 
 the  \emph{product kissing number} \cite{BERB}, or number of worst case occurrences, will grow fast and \emph{a priori} might eat away the product distance gain.
However, this issue seems to be due to the suboptimality of the union bound rather than to the codes themselves.

In fact, let us consider an infinite family of $2k$-dimensional lattices  $L_k \subset \C^k$
with normalized product distance satisfying $(\mathrm{Nd_{p, min}}(L_k))^{2/k}\geq c$, for some positive constant $c$.  

According to Theorem \ref{prop_positive_rate} and Remark \ref{detform} we then have the following.
\begin{corollary}\label{singlecapacity}
Any rate $R$ 
$$
R< \mathbb{E}_h \left[ \log_2 P \abs{h}^2\right] +\log_2 \frac{\pi e}{4}+ \log_2 c,
$$
is achievable with the family $L_k$ of lattices over the fading channel \eqref{single_channel}. 
\end{corollary}

This result proves that indeed we gain by coding over an increasing number of blocks, assuming that we have a family of lattices $L_k$ with the described product distances. 
According to Proposition \ref{motivation}, 
the condition $(\mathrm{Nd_{p, min}}(L_k))^{2/k}\geq c$ implies that
$\mathrm{rh_G}(L_k)\geq kc$. It reveals that families of lattice codes with large product distance do not only have large Hermite invariants, but also that the Hermite invariants of the faded lattices are as large as well. 
Thus, the product distance is not only relevant in capacity considerations or in the high SNR scenario, but also plays a role when coding over a finite number of fading realizations for low SNR.

\subsection{Approaching capacity with number field codes}
Using the normalized product distance as a code design criterion led to lattice constructions based on number fields in \cite{BB,GB, BERB,OV}. 
However, none of these works considered capacity questions.    

Let us now show how  the construction in Proposition   \ref{reg2}, when specialized to the single antenna case, is just   the standard method used  to build lattice codes from number fields \cite{BERB} and how this method can be used to approach capacity in fast fading channels.

Let $K/\Q$ be  a totally complex extension of degree $2k$ and $\{\sigma_1,\dots,\sigma_k\}$ be a set  of  $\Q$-embeddings, such that we have chosen one from each complex conjugate pair. Then we can define a
\emph{relative canonical embedding} of $K$ into $\C^n$ by
$$
\varphi(x)=(\sigma_1(x),\dots, \sigma_k(x)).
$$
The ring of algebraic integers $\OO_K$ has a  $\ZZ$-basis $W=\{w_1,\dots ,w_{2k}\}$ and $\varphi(W)$ is a $\ZZ$-basis for the full  lattice $\varphi(\OO_K)$ in $\C^k$.

 Proposition \ref{reg2} now simplifies to the following.
\begin{corollary}\label{regnum}
Let $\varphi$ be the previously defined embedding and $K$ a degree $2k$ totally complex number field.
Then $\varphi(\OO_K)$ is a $2k$-dimensional lattice in $\C^k$ which satisfies
$$
\mindet{\varphi(\OO_K)}= 1, \,\, \Vol(\varphi(\OO_K))= 2^{-k}\sqrt{|d_K|}
$$
and
$$
\delta(\varphi(\OO_K))=\left(\frac{2^{2k}}{|d_K|}\right)^{1/4}.
$$
\end{corollary}

Using Martinet's family of fields $K_k$ from Theorem \ref{Martinet_theorem} and setting $L_{1,k}=\varphi(\OO_{K_k})$ we have
\begin{align*}
&\Vol(L_{1,k}) \leq\left(\frac{G}{2}\right)^{k} \,\mathrm{and}\,\, {\det}_{min}(L_{1,k})=1,
\end{align*}
where $G \approx 92.368$.
Specializing  to the case where the fading process is i.i.d complex Gaussian we have that any rate
\begin{equation} \label{achievable_rate_SISO}
 R < \log_2(Pe^{-\gamma}) - \log_2\left(\frac{2G}{\pi e}\right),
 \end{equation}
where  $e^{-\gamma}= E_h [\log|h_i|^2]$,   is achievable.


\subsection{Known bounds on discriminants and Hermite invariants}
Equation (\ref{achievable_rate_SISO}) reveals that the codes based on the Martinet family have  a rather large gap to capacity. However, the right-hand side of (\ref{achievable_rate_SISO}) is just a lower bound on the maximum achievable rate with lattice codes, and might be improved with a better error probability estimate and/or a better choice of the lattice sequence. 

An upper bound for the maximum achievable rate using this approach can be derived from a lower bound for the discriminant.
The Odlyzko bound \cite{Odlyzko} states that when $k\rightarrow\infty$ we have that  $|d_K|^{1/2k}\geq 22.3$. If it were possible to reach this lower bound with an ensemble of lattice codes, then any rate $R$ satisfying
\begin{equation} \label{Odlyzko_rate}
R < \log_2(Pe^{-\gamma}) - \log_2\left(\frac{44.6}{\pi e}\right),
\end{equation}
would be achievable. For small values of $k$, there exist number fields having considerably smaller root discriminants.
Table \ref{numbertable} \cite{Odlyzko} lists the best known root discriminants for totally complex number fields of degree
$2k$. The first four values are known to be optimal.

\begin{table}[htbp]
\caption{Best known root discriminants for totally complex number fields $K$ of small degree $2k$.}
\label{numbertable}
\begin{center}
\begin{normalsize}
\begin{tabular}[b]{|c|c|}
\hline
 $k$ &$|d_K|^{1/2k}$ \\
\hline
1&$1.732.. $ \\
\hline
2&$3.289..$  \\
\hline
3&$4.622..$  \\
\hline
4&$5.787..$  \\
\hline
5&$6.793..$ \\
\hline
\end{tabular}
\end{normalsize}
\end{center}
\end{table}

We note that even equation (\ref{Odlyzko_rate}) does not represent an absolute limit for the rates that are achievable with lattice codes and not even with algebraic lattices arising from number fields, and does not mean that the  performance of algebraic codes will always be bounded away from capacity. In fact, as seen in Corollary \ref{singlecapacity}, we are only interested in the normalized product distance of the lattices under consideration. 
For example, instead of considering the image of the ring of integers $\OO_K$ under the embedding $\varphi$, one can use an ideal of this ring of integers (see \cite{Oggier, FOV,ISIT2015_SISO}) or more generally any lattices with good normalized product distance. 

The Minkowski-Hlawka theorem provides a non-constructive proof of the existence of $2k$-dimensional lattices $L_k\subset \C^k$ having Hermite invariants $\mathrm{h}(L_k)\sim \frac{k}{\pi e}$ \cite{CS}. If it were possible to obtain also $h_G(L_k)\sim \frac{k}{\pi e}$ or equivalently  $(\mathrm{Nd_{p, min}}(L_k)) \sim \left( \frac{1}{\pi e}\right)^{k/2}$, then all rates satisfying
\begin{align}\label{gap}
&R< \mathbb{E}_h \left[ \log_2 P \abs{h}^2\right] -\log_2 \frac{4}{\pi e}+ \log_2 \frac{1}{\pi e}= \notag \\
& =\mathbb{E}_h \left[ \log_2 P \abs{h}^2\right] -2
\end{align}
would be achievable with this family of lattices. However, we do not know if this is possible.  The two bit gap  to $\mathbb{E}_h \big[ \log_2 P \abs{h}^2\big]$ in  \eqref{gap} would be exactly the same that is obtained in the AWGN case when using the hard sphere packing approach \cite[Chapter 3]{CS}. Still this two bit gap is not  a fundamental limit of the performance of lattice codes but likely an artifact of the suboptimal method to analyze the error.

\section{Geometry of numbers for fading channels}\label{fading_geom}
In the previous sections, we have shown that the normalized minimum determinant provides a design criterion to build capacity-approaching lattice codes for block fading multiple antenna channels. 
Let us now see how this approach fits into a more general context and can be regarded as a natural generalization of the classical theory of lattices for Gaussian channels.
Finally we show how the code design problems, both in Gaussian and fading channels, can be seen as instances of the same problem in the mathematical theory of geometry of numbers.

Consider a lattice $L \subset \C^k$ having fundamental parallelotope of volume one and define a function
$f_1:\C^k\to \R$ by 
\begin{equation}\label{euclidean}
f_1(x_1,\dots,x_k)=|x_1|^2+|x_2|^2+\cdots+ |x_k|^2.
\end{equation} 
The real number $\mathrm{h}(L)=\inf_{x \in L, \, x\neq{\bf 0}} f_1(x)$ is  
then the Hermite invariant of the lattice $L$.  Let us now denote with
$\mathcal{L}_k$ the set of all $2k$-dimensional lattices in $\C^k$ with volume one.

Suppose that we have an infinite family of lattices  $L_k\in \mathcal{L}_k$ with Hermite invariants satisfying $\frac{\mathrm{h}(L_k)}{k}\geq c$, for some positive constant $c$. As stated in the beginning of Section \ref{general_channels_section}, then all rates satisfying
$$
R < \log_2(P) - \log_2\left(\frac{4}{\pi e}\right)+\log_2 c,
$$
are achievable in the complex Gaussian channel, with this family of lattices. The Hermite invariant $\mathrm{h}(L_k)$ now roughly describes the performance of the lattice $L_k$ and can be used to estimate how close to the capacity a family of lattices can get. This relation is one of the key connections between the theory of lattices and information theory \cite{CS} and has sparked a remarkable amount of research.

\smallskip

Let us now see how  our results 
can be seen as natural generalizations of the relation between Hermite invariant and capacity.

Let us consider $2k$-dimensional lattices $L\subset \C^k$ in $\mathcal{L}_k$ and the form
\begin{equation}\label{fastfading}
f_2(x_1,x_2,\dots, x_k)=|x_1x_2\cdots x_k|.
\end{equation}
Then   $\mathrm{Nd_{p, min}}(L)=\inf_{x\in L, x\neq{\bf 0}} f_2(x) $, is the normalized product distance of the lattice $L$.

Assume that we have an infinite family of lattices  $L_k\in \mathcal{L}_k$ with normalized product distance satisfying $(\mathrm{Nd_{p, min}}(L_n))^{2/k}\geq c$, for some positive constant $c$. As seen before we have that all rates satisfying
$$R< \mathbb{E}_h \left[ \log_2 P \abs{h}^2\right] -\log_2 \frac{4}{\pi e}+ \log_2 c$$
are accessible with this family of lattices with zero error probability over the Rayleigh fast fading channel.

\smallskip

We denote with $\mathcal{L}_{(n,k)}$ the set of all $2n^2k$-dimensional lattices in the space $M_{n \times nk}(\C)$ with volume one. Given $(X_1, X_2,\dots, X_k) \in M_{n\times kn}(\C)$, we consider the function
$$
f_3(X_1, X_2,\dots, X_k)=\prod_{i=1}^k |\det(X_i)|.
$$
For $L \in \mathcal{L}_{(n,k)}$, we have  $$\delta(L)=\inf_{X\in L_k, X\neq{\bf 0}} f_3(X).$$

If  $L_k\subset  M_{n\times kn}(\C)$ is a family of lattices with the property that $\delta(L_k)^{2/kn}\geq c$ then according to  Remark \ref{detform}, any rate satisfying
$$
R< \mathbb{E}_H \left[ \log_2 \det \frac{P}{n} H^{\dagger}H\right] +n\log_2 \frac{\pi e}{4n}+n \log_2 c
$$
is achievable with the lattices $L_k$.

We can now see that the normalized minimum determinant and product distance can be regarded
as generalizations of the Hermite invariant which characterize the gap to capacity achievable with a certain family of lattice codes. 
\smallskip

A natural question is how close to capacity we can get with these methods by taking the best possible lattice sequences.

The \emph{Hermite constant} $H(k)$ can now be defined as 
\begin{equation}\label{f1}
H(2k)=\mathrm{sup}\{ \mathrm{h}(L) \mid L \in \mathcal{L}_{(1,k)}\}.
\end{equation}
In the same manner we can define
\begin{equation}\label{f2}
\mathrm{Nd_{p, min}}(k)=\mathrm{sup}\{\mathrm{Nd_{p, min}}(L) \mid L \in \mathcal{L}_{(1,k)}\}.
\end{equation} 
and
\begin{equation}\label{f3}
\delta(k,n)=\mathrm{sup}\{\delta(L) \mid L \in \mathcal{L}_{(n,k)}\}.
\end{equation}
Each of these constants now represents how close to capacity our methods can take us. Any asymptotic lower bound with respect to $k$ will immediately provide a lower bound for the achievable rate. Just as well upper bounds will give upper bounds for the rates that are approachable with this method.

\smallskip

The characterizations of achievable rates using lattice codes have now been transformed into purely geometrical questions about the existence of lattices with certain properties.
The value of the Hermite constant $H(k)$, for different values of $k$,  has been studied in mathematics for hundreds of years and there exists an extensive literature on the topic. In particular good  upper and lower bounds are available and it has been proven that we can get quite close to Gaussian capacity with this approach \cite[Chapter 3]{CS}. \\
In the case of the product distance, this problem has been considered in the context of algebraic number fields and some upper bounds have been provided. 
As far as we know the best lower bounds come from the existence results provided by number field constructions \cite{Xing} and \cite{ISIT2015_SISO}.

The 
properties 
of $\delta(k,n)$ 
have been far less researched in the literature. Simple upper bounds can be derived from bounds for Hermite constants as pointed out in \cite{LV}
and lower bounds are 
obtained from division algebra constructions as described in this paper, but  
the mathematical literature doesn't seem to offer any ready-made results for this problem. \\
However, all three questions can be seen as special cases of the problem of finding the minimum of a homogeneous form over a lattice in the mathematical theory of \emph{geometry of numbers} \cite{GL}. Let us now elaborate on the topic.

\begin{definition}
A continuous 
function 
$F$: $M_{n\times kn}(\C) \to \R$
is called a homogeneous form of degree $\sigma>0$ if it satisfies the relation
$$
|F(\alpha {X})|=|\alpha|^{\sigma} |F(X)|\quad (\forall \alpha \in \R, \forall X \in M_{n\times kn}(\C)).
$$
\end{definition}

Let us consider the body $S(F)=\{X \,|\,X \in M_{n\times kn}(\C), |F(X)|\leq 1\}$, and a 
$2kn^2$ dimensional lattice $L$ with a  fundamental parallelotope of volume one.

We  then  define the \emph{homogeneous minima} $\lambda(F,L)$ of $F$ with respect to the lattice $L$ by
$$
\lambda(F,L )=(\mathrm{inf}\{\lambda|\,\lambda>0, \mathrm{dim}(\R(\lambda S(F)\cap L))\geq 1\})^{\sigma},
$$
where $\R(\lambda S(F)\cap L)$ is the $\R$-linear space generated by  the elements in $\lambda S(F)\cap L$.
This allows us to define the \emph{absolute homogeneous minimum} 
$$
\lambda(F)=\sup_{\Vol(L)=1}\lambda(F,L).
$$

We can now see that all of our forms $f_1$, $f_2$ and $f_3$ are homogeneous forms. For the  Hermite invariant we have $\sigma=2$, for the product distance $\sigma=n$, and for the normalized minimum determinant $\sigma=n^2k$. We can also easily see that the constants \eqref{f1}, \eqref{f2} and \eqref{f3}
are absolute homogeneous minima of the corresponding forms.

These results suggest that there is a very general connection between information theory and geometry of numbers for different channel models. 
It seems that given a fading channel model, there exists a form whose absolute homogeneous minima provide a lower bound for the achievable rate using lattice codes.

\begin{remark}
The definitions for the geometry of numbers given in this section were stated for lattices in the space $M_{n\times nk}(\C)$, while usually the definitions are given in the space $\R^m$. This is however, just to keep 
our notation simple. 
The space $M_{n\times nk}(\C)$ can be identified with the space $\R^{2n^2k}$ and we could have given the definitions also in the traditional form using this identification.

\end{remark}

\section{Discussion and questions for further research} \label{conclusion}

 In this work we proved the existence of lattice codes achieving constant gap to capacity in ergodic fading channels. Unlike the case of existence results based on random coding, our finite codes are always built from the same family of lattices, irrespective of the SNR and even of the fading statistics. Hence, using the minimum determinant as a  design principle leads to extremely robust codes. In particular division algebra and number field codes have this robustness property.

However, our codes still have a considerable gap to capacity and further research is needed. Let us now point out a few directions this research can take next.

 In the case of single user channels the clearest goal is to improve our methods and close the gap to capacity. We note that this gap depends on several factors. First of all, the normalized minimum determinant affects the value of the gap. Second, our bound for the error probability is based on sphere packing and thus is suboptimal.  
 
Thus, the possible improvements to our construction are two-fold. In the first place, one could try to find families of lattices $L_{n,k}\subset M_{n,nk}(\C)$ with larger normalized minimum determinant, for instance by replacing the  
centers in our constructions with families of number fields having smaller discriminants. One can also consider more general examples of lattices than those arising from orders in division algebras: for example ideals of orders, or in the case of number field codes, ideals of the ring of algebraic integers. 
In the second place, in this paper we have not considered the issue of shaping. Improving the shaping properties of our lattices might lead to a better error probability bound.

Another approach is to relax  our minimum determinant code design criterion. Our codes are extremely robust and quite universal in the sense that they respond very well to any non pathological fading realization. This universality is of course a strength, but it could also lead to a situation where the codes are rather good for every channel, but not optimal for any. If we fix a channel model, it may be possible to 
weaken the design principle. 
This might allow us to consider larger ensembles of lattices and possibly to close the gap to capacity in this fixed channel model.

Our work was about explicit code constructions in the spirit of classical sphere packings \cite{CS}. However, it seems that even the existence of capacity achieving lattice codes in fading channels is an open question (see Section \ref{related}). 
In the case of AWGN channel  this question was solved only quite recently in \cite{Erez_Zamir} by assuming that the receiver and transmitter have access to a common source of randomness.  One of the key elements in the achievability result of \cite{Erez_Zamir} 
is the  Minkowski-Hlawka theorem, that can be used to prove the existence of lattices with certain properties. 
It is then a natural idea to generalize this approach by proving an analogue of the Minkowski-Hlawka theorem for fading channels. It seems to us that this problem is non-trivial.

In this paper we have considered block fading MIMO channels, but we hope that the methods developed here can be applied also in a more general setting. Let us now sketch an outline for possible generalizations.

The reduced Hermite invariant is a natural analogue of the classical Hermite invariant for fading channels. This concept can likely be generalized to other fading channel models, such as for example intersymbol interference channels. 
Given a fading channel we can ask what would be  the group (or set) $G$  that would represent the action of the channel, 
and define the corresponding reduced Hermite invariant $h_G$.  The next question is then to find lattices that would maximize this value.  
In  the case of the block fading channel, the problem was made more accessible by  Proposition \ref{motivation}, where we proved that $h_G$ can be seen as the minimum of  a certain homogeneous form. 
This line of thought suggests a general approach to turn the chase for capacity into a problem in geometry of numbers for different channel models. It also raises several questions.  For example we can ask which are the channel models where this approach can be applied and for which groups $G$ the reduced Hermite invariant corresponds to some homogeneous form.

Finally, the lattice codes proposed in this paper could have applications to other problems in information theory, such as coding for multiple access fading channels and for information theoretic security.


\renewcommand{\thesubsection}{\Alph{subsection}}
\appendix

\subsection{Proof of Theorem \ref{prop_positive_rate2}} \label{proof_positive_rate2}
With a similar approach as in the proof of Theorem \ref{prop_positive_rate}, we consider the following upper bound:
\begin{align}
& P_e \leq \mathbb{P}\Bigg\{ \norm{W}^2 \geq \left(\frac{d_{H}}{2}\right)^2\Bigg\} \leq \notag \\
& \leq \mathbb{P}\left\{ \frac{\norm{W}^2}{knn_r} \geq 1+\epsilon\right\} + \mathbb{P}\left\{ \frac{d_H^2}{4knn_r} < 1+\epsilon\right\} \label{sum2}
\end{align}
The first term in equation (\ref{sum2}) tends to zero exponentially fast when $k \to \infty$ since $2 \norm{W}^2 \sim \chi^2(2knn_r)$. We now focus on the second term in equation (\ref{sum2}), and begin by finding a lower bound on the minimum distance $d_H$ in the received constellation. \\
For all $i \in\{1,\ldots,k\}$, let $\lambda_{i,j}$, $j=1,\ldots,n$ be the singular values of $H_i^{\dagger}H_i$ with
$$ 0=\lambda_{i,1}=\cdots=\lambda_{i,n-n_r} < \lambda_{i,n-n_r+1} \leq \cdots \leq \lambda_{i,n},$$ 
and $l_{i,j}$ the singular values of $X_iX_i^{\dagger}$ with $$l_{i,1} \geq l_{i,2} \geq \cdots \geq l_{i,n}.$$
Using the mismatched eigenvalue bound \cite{Kose_Wesel,EKPKL}, we have
$$ \norm{H_i X_i}^2 \geq \sum_{j=1}^n \lambda_{i,j} l_{i,j}=\sum_{j=n-n_r+1}^n \lambda_{i,j} l_{i,j}.$$
Consequently, we find that
\begin{align}
&d_H^2 \geq \alpha^2 \min_{X \in L_{n,k} \setminus \{0\}} \sum_{i=1}^k \sum_{j=n-n_r+1}^n \lambda_{i,j} l_{i,j} \geq \notag \\
&\geq \alpha^2 n_r k \prod_{i=1}^k \prod_{j=n-n_r+1}^n (\lambda_{i,j} l_{i,j})^\frac{1}{n_r k} \label{d_H_bound2}
\end{align}
Using the NVD property of the code, we get
$$\prod_{i=1}^k \prod_{j=1}^{n} l_{i,j}=\prod_{i=1}^k \abs{\det X_i}^2 \geq 1$$
Therefore, we have the lower bound
{\allowdisplaybreaks
\begin{align*}
& \prod_{i=1}^k \prod_{j=n-n_r+1}^{n} l_{i,j}\geq \Bigg(\prod_{i=1}^k \prod_{j=1}^{n-n_r} l_{i,j}\Bigg)^{-1} \geq \\
&\geq \Bigg(\frac{1}{(n-n_r)k} \sum_{i=1}^k \sum_{j=1}^{n-n_r} l_{i,j}\Bigg)^{-(n-n_r)k} \geq \\
& \geq \left(\frac{Pn}{\alpha^2(n-n_r)}\right)^{-k(n-n_r)}
\end{align*}
}%
where we have used the arithmetic-geometric mean inequality and the power constraint $\alpha^2 \norm{X}^2=\alpha^2 \sum_{j=1}^n l_{i,j} \leq Pkn$. Replacing the previous expression in (\ref{d_H_bound2}), we obtain
\begin{align*}
& d_H^2 \geq \frac{\alpha^2 n_r k \prod_{i=1}^k \prod_{j=n-n_r+1}^n \lambda_{i,j}^\frac{1}{n_r k}}{\left(\frac{Pn}{\alpha^2(n-n_r)}\right)^{\frac{n-n_r}{n_r}}}=\\
&=(\alpha^2)^{\frac{n}{n_r}} \left(\frac{n-n_r}{Pn}\right)^{\frac{n-n_r}{n_r}}  n_r k \prod_{i=1}^k \det(H_iH_i^{\dagger})^\frac{1}{n_r k}
\end{align*}
The second term in (\ref{sum2}) can thus be upper bounded by 
\begin{align*} 
&\mathbb{P}\Bigg\{ \prod_{i=1}^k \det(H_iH_i^{\dagger})^\frac{1}{n_r k} < 4(1+\epsilon) \left(\frac{n}{\alpha^2}\right)^{\frac{n}{n_r}}\Big(\frac{P}{n-n_r}\Big)^{\frac{n-n_r}{n_r}} \Bigg\}\\
&=\mathbb{P}\Bigg\{ \frac{1}{k} \sum_{i=1}^k \log \det H_i H_i^{\dagger} < \log \frac{(4(1+\epsilon))^{n_r} n^nP^{n-n_r}}{\alpha^{2n}(n-n_r)^{n-n_r}}\Bigg\}
\end{align*}
By hypothesis the weak law of large numbers (\ref{WLLN2}) holds, i.e. $\frac{1}{k} \sum_{i=1}^k \log \det H_iH_i^{\dagger} \to \mu$ as $k \to \infty$. Thus, the error probability will vanish provided that for sufficiently large $k$, 
$$\log \frac{(4(1+\epsilon))^{n_r} n^nP^{n-n_r}}{\alpha^{2n}(n-n_r)^{n-n_r}} < \mu$$
Recalling that $\alpha^2 \geq\frac{C_{n,k}^{\frac{1}{n^2k}}P}{2^{\frac{R}{n}}C_L}$, the condition can be rewritten as
\begin{multline*}
R< \mu +n_r \log P- n_r\log 4(1+\epsilon) -n \log n C_L + \\ + \frac{\log C_{n,k}}{nk} +(n-n_r)\log(n-n_r)
\end{multline*}
Using Stirling's approximation (\ref{C_Stirling}), for large $k$ we have
$$\frac{\log C_{n,k}}{nk} \approx  n \log \pi e - n \log n -\frac{1}{2nk}\log 2\pi n^2 k $$
Asymptotically, we find that any rate
$$R< \mu +n_r \log \frac{P}{4(1+\epsilon)} -n \log \frac{n^2 C_L}{\pi e} +(n-n_r)\log(n-n_r)$$
is achievable. Since this is true for all $\epsilon>0$, this concludes the proof. \hspace*{\fill}~\IEEEQED

\begin{small}

\end{small}

\begin{IEEEbiographynophoto}{Laura Luzzi} received the degree (Laurea) in
Mathematics from the University of Pisa, Italy, in 2003 and the Ph.D.
degree in Mathematics for Technology and Industrial Applications from
Scuola Normale Superiore, Pisa, Italy, in 2007. From 2007 to 2012 she held
postdoctoral positions in T\'el\'ecom-ParisTech and Sup\'elec, France, and
a Marie Curie IEF Fellowship at Imperial College London, United Kingdom.
She is currently an Assistant Professor at ENSEA de Cergy, Cergy-Pontoise,
France, and a researcher at ETIS (ENSEA - Universit\'e de
Cergy-Pontoise- CNRS).\\
Her research interests include algebraic space-time coding and decoding
for wireless communications and physical layer security.
\end{IEEEbiographynophoto}

\begin{IEEEbiographynophoto}{Roope Vehkalahti}
received the M.Sc. and Ph.D. degrees from the University of Turku, Finland,
in 2003 and 2008, respectively, both in pure mathematics.

Since September 2003, he has been with the Department of
Mathematics, University of Turku, Finland.  In 2011-2012 he was visiting Swiss Federal Institute of Technology, Lausanne (EPFL).  His  research interest include  applications of algebra and number theory to information theory.
\end{IEEEbiographynophoto}

\end{document}